%% file: paper2.tex
\documentclass[conference,10pt,draftclsnofoot,onecolumn]{IEEEtran}
\IEEEoverridecommandlockouts
\usepackage[dvips]{graphicx}
\usepackage[cmex10]{amsmath}
\usepackage{amsfonts,amssymb,bm}
\usepackage[tight,footnotesize]{subfigure}
\usepackage{fixltx2e}
\usepackage{dblfloatfix}
\usepackage{rotating}
\usepackage{multirow}
\usepackage{url}
\usepackage{cite}
\usepackage{algorithmic,algorithm}
\usepackage{slashbox}
\usepackage{cite}
\usepackage{setspace}
\usepackage{footnote}
\usepackage[T1]{fontenc}
\usepackage{ae,aecompl}
\usepackage{epsfig}
\usepackage{multicol}
\usepackage{subfigure}
\usepackage{times}
\algsetup{linenodelimiter=.}

\newtheorem{proposition}{Proposition}[section]

\newcommand{\ie}[0]{\textit{i.e.},}

\newcommand{\etal}[0]{\textit{et al.}}

\newcommand{\vs}[0]{\textit{vs.}}

\begin{document}

\title{Scalable distributed service migration via Complex Networks Analysis}
\author{\IEEEauthorblockN{\large Panagiotis Pantazopoulos\IEEEauthorrefmark{1} \hspace{20pt} Merkourios Karaliopoulos\IEEEauthorrefmark{2} }
\\
\IEEEauthorblockA{\IEEEauthorrefmark{1}
 Department of Informatics and Telecommunications\\
 National \& Kapodistrian University of Athens\\
 Ilissia, 157 84 Athens, Greece\\
 Email: \{ppantaz, ioannis\}@di.uoa.gr}
\and
\IEEEauthorblockN{\large \hspace{18pt}  and \hspace{20pt} Ioannis Stavrakakis\IEEEauthorrefmark{1}}
\\
\IEEEauthorblockA{\IEEEauthorrefmark{2}
Department of Information Technology \\
and Electrical Engineering\\
ETH Zurich, Zurich, Switzerland  \\
 Email: karaliom@tik.ee.ethz.ch}

% 
% \title{Scalable distributed service migration via Complex Networks Analysis}
% \author{\IEEEauthorblockN{Panagiotis Pantazopoulos \hspace{21pt} Merkourios Karaliopoulos \hspace{15pt} and \hspace{15pt} Ioannis Stavrakakis}\\
% \IEEEauthorblockA{Department of Informatics and Telecommunications\\
% National \& Kapodistrian University of Athens\\
% Ilissia, 157 84 Athens, Greece\\
% Email: \{ppantaz, mkaralio, ioannis\}@di.uoa.gr}

\thanks{This work has been supported by the IST-FET project SOCIALNETS (FP7-IST-217141).}}

\maketitle
\begin{abstract}

With social networking sites providing increasingly richer
context, User-Centric Service (UCS) creation is expected to
explode following a similar success path to User-Generated
Content. One of the major challenges in this emerging highly
user-centric networking paradigm is how to make these exploding in
numbers yet, individually, of vanishing demand services available
in a cost-effective manner. Of prime importance to the latter (and
focus of this paper) is the determination of the optimal location
for hosting a UCS. Taking into account the particular
characteristics of UCS, we formulate the problem as a facility
location problem and devise a distributed and highly scalable
heuristic solution to it.

Key to the proposed approach is the introduction of a novel metric
drawing on Complex Network Analysis. Given a current location of
UCS, this metric helps to a) identify a small subgraph of nodes
with high capacity to act as service demand concentrators; b)
project on them a reduced yet accurate view of the global demand
distribution that preserves the key attraction forces on UCS; and,
ultimately, c) pave the service migration path towards its optimal
location in the network. The proposed iterative UCS migration
algorithm, called cDSMA, is extensively evaluated over synthetic
and real-world network topologies. Our results show that cDSMA
achieves high accuracy, fast convergence, remarkable insensitivity
to the size and diameter of the network and resilience to
inaccurate estimates of demands for UCS across the network. It is
also shown to clearly outperform local-search heuristics for
service migration that constrain the subgraph to the immediate
neighbourhood of the node currently hosting UCS.

\end{abstract}

\input{introduction}

\input{problem}

\input{bc}

\input{analysis}

\input{evaluation}

\input{results}

\input{related}

\input{discussion}
\section*{Acknowledgements}

We would like to thank Andrea Passarella from CNR-Pisa and Xenofontas Dimitropoulos from ETH Zurich for insightful conversations
and comments.

\bibliographystyle{IEEEtran}
\bibliography{paper2_ref}

%If needed

\end{document}

%% file: introduction.tex
\section{Introduction}\label{sec:intro}

% motivate the User-generate service paradigm
One of the most significant changes in networked communications
over the last few years concerns the role of the \emph{end-user}.
Till recently the end-user has been the \emph{consumer} of content
and services generated by explicit entities called content and
service providers, respectively. Nowadays, the Web2.0 technologies
have resulted in/enabled a paradigm shift towards more
user-centric approaches to content generation and provision. The
first strong evidence of this shift has been the abundance of
\emph{User-Generated Content} ($UGC$) in social networking sites,
blogs, wikis, or video distribution sites such as YouTube, which
motivated even the rethinking of the Internet architecture
fundamentals \cite{contnets}, \cite{Jacobson09}. The
generalization of the $UGC$ concept towards services is
increasingly viewed as the next major trend in user-centric
networking \cite{GalisFIA}.

% explain the User-generate service paradigm and its requirements
The user-oriented service creation concept aims at engaging
end-users in the generation and distribution of service
components, more generally \emph{service facilities} \cite{ugs09}.
To facilitate the wide proliferation of the so-called
\emph{User-Generated Service} ($UGS$) paradigm, some key technical
challenges should be addressed such as a) the design of simple
programming interfaces that will enable the involvement of
end-users without strong programming background; and b) the
deployment of scalable distributed mechanisms for discovering,
publishing, and moving service facilities within the network.

Our work focuses on the second challenge. In particular, it
addresses the problem of optimally placing service facilities
within the network so that the cost of accessing and using them is
minimized. The problem is typically viewed as an instance of the
family of \emph{facility location} problems and is formulated as
an $1$- or, more generally, $k$-median problem, depending on
whether facilities can be replicated in the
network~\cite{mirchandani1990}. The main bulk of proposed
solutions to the problem are centralized (see, for
instance,~\cite{JainVazi01}): the optimal service location is
determined by a single entity that possesses global information
for both the network topology and distribution of service demand
across the network. Nevertheless, the service deployment scenarios
considered in this work, involving the flexible and scalable
deployment of many distributed user-generated services within
possibly large networks, bring traditional centralized approaches
to the problem solution to their limits.
%When moving towards feasible
%solutions for (\textit{User-Generated}) service deployment in
%large-scale distributed systems,
%a number of critical limitations concerning the above
%traditional facility location approach, can be identified.
%The $1(k)$-median problem is solved in a centralized manner
%\textbf{[refs]}.
%The globally optimal location of the service is
%determined by a single entity that should possess
%~\footnote{Let alone computation time that
%might turn out to be prohibitive for large-scale networks.}
%explicit global information for both the network topology and
%distribution of service demand across the network.
Gathering the required information to a single physical location
%in a scalable manner
is already a challenge. Furthermore, the centralized treatment of
the problem assumes the existence of an ideal super node bearing
the burden of decision-making.
%introduces an implicit
%logical hierarchy in the role of nodes in the network, placing
%asymmetrically higher burden on the decision-making entity.
%~\footnote{Let alone computation time that
%might turn out to be prohibitive for large-scale networks.}
This burden is primarily due to the computationally intensive $1(k)$-median problem
\cite{mirchandani1990}.
% which becomes a
%difficult task for large-scale networks, as we explain in Section
%\ref{sec:problem}.
Given that (minor) user demand shifts or
network topology changes may alter the optimal service location,
it is neither practical nor affordable to each time centrally
compute a new problem solution.

In our paper we propose solutions for overcoming these
limitations. The approach we have taken is highly decentralized;
%individual nodes (provisionally) hosting
the service facilities migrate from the node that makes them
available towards the minimum-cost location by traversing a
cost-decreasing path. It is also scalable;
%decide locally \emph{whether} the service should stay with them or
%\emph{migrate} to another lower-cost node.
similar to other proposals in literature, it moves the service
facilities in the network by iteratively solving \emph{locally} a
much smaller scale 1(k)-median optimization problem than what a
global centralized solution would require. Nevertheless, it
departs from standard practices in the way it selects the nodes
for the local 1(k)-median problem. State-of-the-art approaches
(for example, \cite{OikonomouSX08, SLOSB-DSMSISD-10}) recruit
those nodes from their immediate local neighborhood. On the
contrary, our algorithm invests additional effort to make a more
informed selection of these nodes, which promotes the ``correct"
directions of migration towards the globally optimal locations.

%We show later in \ref{subsec:positioning} that t
%Topological and demand information still needs to propagate in the
%network and in Section~\ref{sec:discussion} we discuss mechanisms
%that could ease this task.
%However, the service migrates within
%the network over a few computationally cheaper hops that form a
%cost-decreasing path towards its optimal location.
%The computational load is spread amongst the nodes on the migration
%path and consists in the iterative solution of the 1-median
%problem at a much smaller scale than what a global centralized
%solution would require.
To achieve this, we devise a metric, called \emph{weighted
Conditional Betweenness Centrality (wCBC)}, that draws on
\emph{Complex Network Analysis} (CNA). CNA provides a theoretical
framework for unified modeling and analysis of several types of
networks and the expectations in the networking community are that
its insights could benefit the design of more efficient network
protocols. In our work, the CNA-inspired metric helps the service
migrate through the network both faster and towards better service
locations.
%As it will is explained in detail in Sections
%\ref{sec:metric} and \ref{sec:analysis},

\begin{figure}[tbp]
%\label{wons}
\begin{minipage}[b]{.52\linewidth}
\centering
\framebox(200,170){\includegraphics[width=4.5cm]{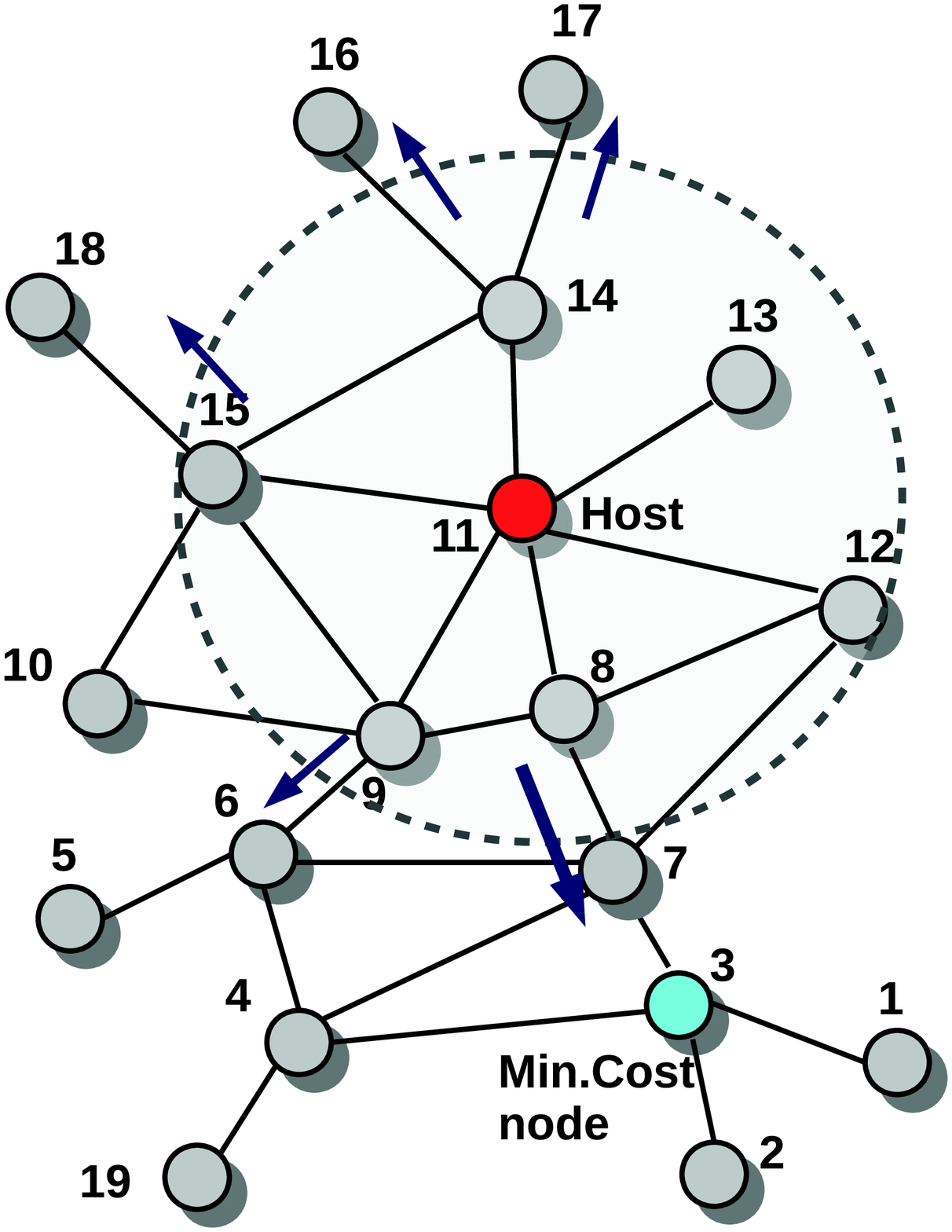}}
%\caption[short]{}
\end{minipage}
\hspace{0.15cm}
\begin{minipage}[b]{0.3\linewidth}
\centering
\framebox(200,170){\includegraphics[width=4.5cm]{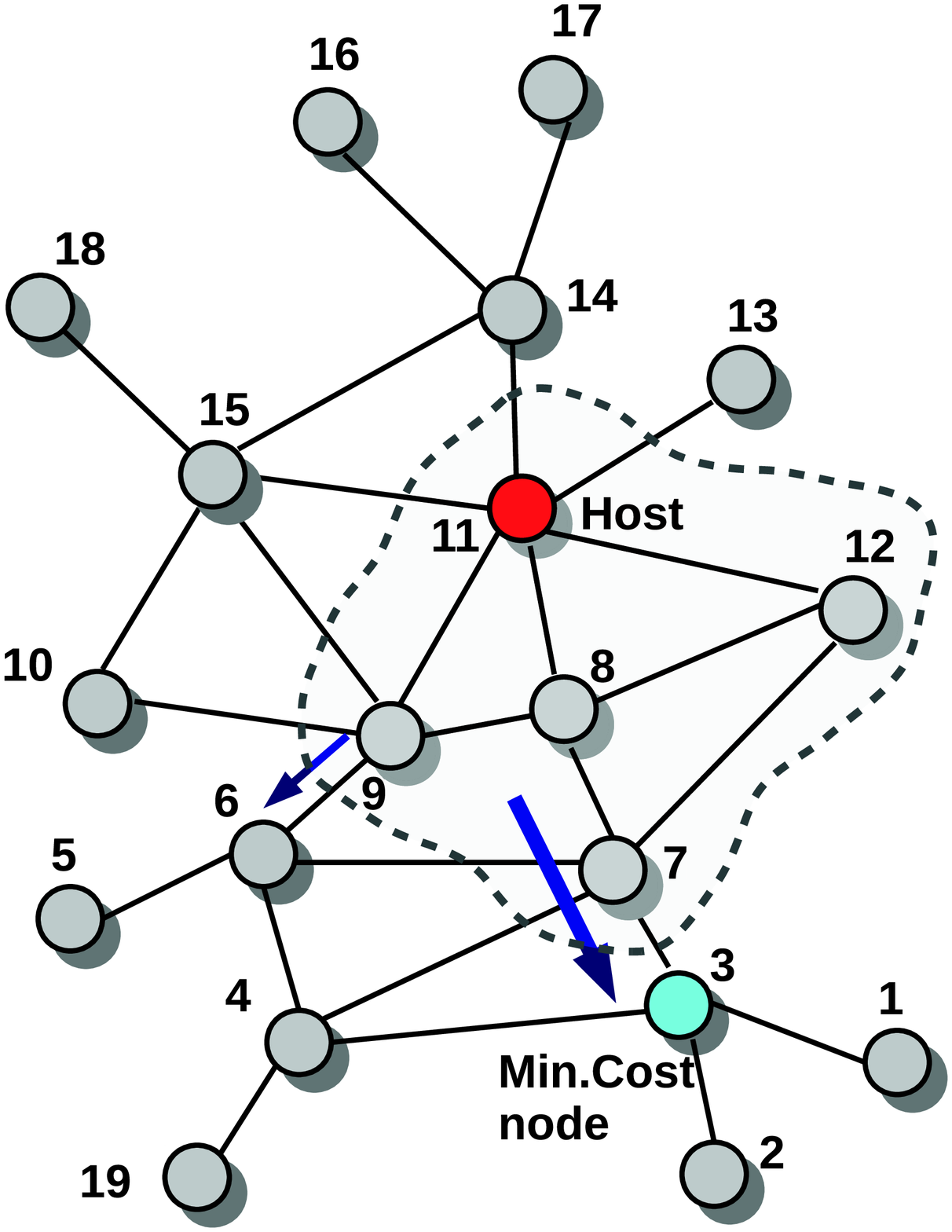}}
%\caption[short]{}
\end{minipage}
\caption{Looking for the next service migration
hop: local-search heuristics(left) \vs~cDSMA(right).}
\label{attraction_plot} \vspace{-0.2 in}
\end{figure}

In each service migration step, the metric serves two purposes.
Firstly, it \emph{identifies} those nodes that contribute most to
the aggregate service access cost and pull the service strongly in
their direction; namely, nodes that hold a central position within
the network topology and/or route large amounts of the demand for
it. Secondly, it correctly \emph{projects} the \emph{attraction
forces} these nodes exert to the service upon the current service
location and facilitates a migration step towards the optimal
location (Fig. \ref{attraction_plot}). It is really tempting to
draw an analogy between our mechanism and a directional antenna:
through the directed search (main radiation lobes) the mechanism
amplifies the impact of the major service demand attractors (signal
sources) while suppressing less important demand poles (noise)
that blind local search (omnidirectional antenna) would induce.

%on the small subset of nodes on the service migration path,
%For given service placement, the metric reflects how much demand
%for the service is routed through each network node. In each
%service migration step, the CNA-driven metric serves two purposes.
%Firstly, nodes with the top $wCBC$ scores are selected as members
%of a significantly smaller subgraph of candidate service hosts,
%wherein the $1$-median problem is solved (\emph{$1$-median
%subgraph}). Secondly, through simple manipulation of the metric,
%the service demand originating from the remaining nodes is
%accurately mapped on the $1$-median subgraph nodes.

%To achieve this, we draw on \emph{Complex Network Analysis} (CNA)
%and devise a metric, called $weighted$ $Conditional$ $Betweenness$
%$Centrality$ $(wCBC)$. For given service placement, the metric
%reflects how much demand for the service is routed through each
%network node. In each service migration step, the CNA-driven
%metric serves two purposes. Firstly, nodes with the top $wCBC$
%scores are selected as members of a significantly smaller subgraph
%of candidate service hosts, wherein the $1$-median problem is
%solved (\emph{$1$-median subgraph}). Secondly, through simple
%manipulation of the metric, the service demand originating from
%the remaining nodes is accurately mapped on the $1$-median
%subgraph nodes.
%This way, the service facilities move towards
%their optimal location through a number of computationally lighter
%steps.

We detail the metric and our algorithm, hereafter called cDSMA, in
Sections \ref{sec:metrics} and \ref{sec:analysis}, respectively,
and evaluate them extensively in Section \ref{sec:analysis}.
When running over real-world ISP topologies, the cDSMA achieves
remarkably high accuracy and fast convergence even when the
$1$-median problem iterations are solved locally with very few
nodes, less than ten. Moreover, the cDSMA performance is
practically invariable with the network size and diameter as well
as the spatial dynamics of the service demand distribution across
the network. Compared with distributed local-search policies,
%where the service selects the next migration hop by iteratively
%solving the $1$-median problem within the local neighborhood of
%its current location,
cDSMA yields consistently better service placements, which are not
dependent on the location of the service generation. We summarize
our findings, discuss practical implementation aspects of our
mechanism, and sketch possible extensions to this work in Section
\ref{sec:discussion}.

% \textit{Contribution}: Herein we develop a
%scalable heuristic approach to cope with the high complexity and
%(topology-demand) knowledge limitations of the distributed-fashion
%optimization of (User Generated) service placement. Our approach
%amounts to iteratively solving the well-known \textit{1-median}
%problem using (each time) a restricted -in size- solution space.
%To do so, we introduce an innovative CNA-based metric with which
%we single out those nodes that are somehow capable of inferring
%the demand load (for the available service) coming from those left
%out of the problem solution. \textbf{This method, tested through
%simulation against various scenarios, is proved to outperform the
%state-of-the-art distributed facility location algorithms}
%\textit{sooth it a little bit!!}.
%
%The rest of the  paper is organized as follows.....

%% file: problem.tex
%\section{Problem Statement and solution requirements}\label{sec:problem}
\section{Service placement: a facility location problem}
\label{sec:problem}
% ===========Merkouris' comments========================================================
% \textbf{Main points:} \\
% a) Formulation of the problem in terms of the facility location problem (1-median variant) \\
% b) Explain the NP-hardness of the problem to motivate the use of
% approximate solutions --summarize the approach (CNA-inspired)
% ======================================================================================

%In the sequel we provide the mathematical notation of one of the
%most studied problems (see section~\ref{Related_work} for a
%variety of approaches) in operations research, formulated to
%address the service placement problem.

% is the facility location problem where there is.... ..
% Such optimization problems are computationally challenging, if not
% practically impossible to solve (for general topologies the
% \textit{k-median} is NP-hard~\cite{kariv1979}), often requiring
% heuristics for limiting the size of the solution space.

%\subsection{Service placement as a facility location problem}

The optimal placement of service facilities in some network
structure has been typically tackled as an instance of the
facility location problem~\cite{mirchandani1990}. Input to the
problem is the topology of the network nodes that may host
services and/or network users. The objective is to place services
in a way that minimizes the aggregate cost of accessing them over
all network users.
%The main
%question here is to identify the network nodes of a given topology
%that can host some service, so that the latter is accessed by the
%network users residing at the various network nodes with the
%minimum average cost.

More precisely, the network topology is represented by an
undirected connected graph $G(V,E)$, where $V$ is the set of nodes
and $E$ is the set of edges(links) connecting them. Without loss
of generality, we assume that all links have a unit weight and
thus the minimum cost path $d(n,m)$ between nodes $n$ and $m$,
corresponds to the minimum hop count path linking $n$ and $m$.
Each network node $n$ serves users that access the service with
different intensity(frequency), generating an aggregate demand
$w_{n}$ for the service.
% according to the preferences of users
%behind it.
When there are $k$ service facilities available, the
problem of their optimal placement in the network can be
formulated as the classical \textit{k-median} problem; namely, the
set $\mathcal{F}$ of $k$ nodes ($|\mathcal{F}|=k$) that are
selected to host service facilities minimize the aggregate service
access cost:

\begin{equation}\label{eqn:k-median}
%Cost(\mathcal{F}) =\sum_{n\in\mathcal{V}} w(n) \cdot min_{x_{j}
%\in \mathcal{F}}\{d(x_{j},n)\}
Cost(\mathcal{F}) =\sum_{n\in\mathcal{V}} w(n) \cdot
\underset{x_{j} \in \mathcal{F}} {\operatorname{min}}
\{d(x_{j},n)\}
\end{equation}
where $min\{d(x_{j},n)\}$ denotes the distance between each node
$n$ and its closest service host node $x_{j} \in \mathcal{F}$. In
this paper we focus our attention to the single service facility
scenario, $|\mathcal{F}|=1$. Practically, service facilities
migrate across the whole network seeking their optimal
placement without the possibility for replication. The $1$-median
formulation matches better the expectations about the
User-Generated Service paradigm, \ie many different services
generated in various places in the network raising 
small-scale interest so that replication of their facilities be
less attractive.
%of relatively
%low demand and seeks to find its optimal placement within a
%network.
The respective 1-median problem formulation, minimizing the
access cost of a service located at node $k\in V$
%and a cost to be minimized
is given by:
\begin{equation}\label{eqn:1-median}
 Cost(\mathcal{F}) =\sum_{n\in\mathcal{V}} w(n) \cdot d(k,n)
\end{equation}
In general topologies, optimization problems such as 1-median and
$k$-median, are NP-hard\footnote{While in special cases like the
tree topology with equal link weights, the 1-median may
need $\mathcal{O} ({|E|}^2) $ time to solve (using exhaustive
search~\cite{kariv1979} or even faster for more efficient
algorithms~\cite{Goldman}), those problems are typically
characterized by the computationally difficult case.} requiring
global information about the network topology and generated demand
load~\cite{kariv1979}. Thus, so far the main bulk of relevant
theoretical work is in the field of approximation algorithms,
where various techniques have been applied~\cite{alg_k_Med}.

\subsection{Exploiting CNA to overcome limitations}

We make use of Complex Network Analysis (CNA) to dramatically
reduce the scale of the 1-median problem that network nodes need
to solve. We introduce a metric, called Weighted Conditional
Betweenness Centrality (wCBC), which draws on a known CNA metric
(Betweenness Centrality~\cite{citeulike:392801}) and assesses the
value of nodes as candidate hosts of the service. The nodes with
the highest $wCBC$ values induce a small subgraph on the original
network graph, wherein the original optimization problem can be
solved more efficiently for the next-best service location in the
network. Besides identifying the highest-value nodes for hosting
the service, the metric directly lets us map the demand of the
rest of the network nodes on this subgraph. We introduce our
metric in Section \ref{sec:metrics}, whereas the demand mapping
process and the overall algorithm are detailed in Section
\ref{sec:analysis}.

%To address the placement problem in a distributed fashion we argue
%that the identification of present structures (i.e. subgraph) of
%``significant'' nodes -through CNA analysis- can help us harness
%the computationally intensive task of solving the large
%optimization problem; as such, we restrict the -otherwise- global
%optimization to a  (small-numbered) sequence of small-scale ones
%(each of them solved within a subgraph) and, thus, provide a
%scalable approach.

%For the needs of such a ``local-search-like'' approach,  we should
%somehow be aware of, or otherwise induce the demand (for the
%service) stemming from the rest of the network's nodes, not
%included in the current optimization's area.

%A similar decentralized approach was taken by Smaragdakis \etal
%in~\cite{SLOSB-DSMSISD-10}. Compared to our approach, the authors
%were practically fixing a-priori the feasible solution space for
%the local optimization problems to the R-hop neighborhood of the
%current service locations -they were considering the k-median
%problem variant. Small values of R, in the order of 1 or 2, ease
%information gathering but slow down the migration process. On the
%contrary, our metric exploits CNA insights to naturally expand the
%search area for the local optimization problem and drive us faster
%towards the optimal service location. Notably, the service
%migration time to the optimal location is largely independent of
%the initial service placement, \ie the node first deploying the
%service. We elaborate on how the approaches relate in Section
%\ref{subsec:positioning}.
A similar decentralized approach to the service placement problem
was taken by Smaragdakis \etal in~\cite{SLOSB-DSMSISD-10},
although they allow for service replication. Compared to our
approach, the authors practically fix a-priori the $k$ $1$-median
subgraphs to the $R$-hop neighborhood of the current service
locations. Intuitively, small values of $R$, in the order of one
or two, ease information gathering but slow down the migration
process. We show later in Section \ref{subsec:positioning} that
this local-search approach \emph{also} adds ``noise'' to the
algorithm's effort to push the service towards the optimal
location in the network. As a result, the algorithm may be trapped
in locally anticipated as optimal, yet globally suboptimal,
locations. On the contrary, our metric exploits CNA insights to
naturally extract a more informed 1-median subgraph and drive us
faster towards the optimal service location.
%Notably, the service
%migration time to the optimal location is largely independent of
%the initial service placement, \ie the node first deploying the
%service.
We elaborate on how the two approaches relate in Section
\ref{subsec:positioning}.

% The fairness argument needs some discussion, depends how fairness is defined.

%The former approach (i.e., demand awareness) applied by
%Smaragdakis et al.~\cite{SLOSB-DSMSISD-10} involves the local
%gathering of the demand load statistics that reach the ring of the
%optimization's area. Apart from some overhead
%induced~\footnote{Extra cost paid for each iteration's gathering
%of statistics and for computation of the demand mapping
%error~\cite{SLOSB-DSMSISD-10}.} we claim (see \ref{positioning})
%that the latter approach (i.e demand inference) -explored here-
%seems preferable, at least when dealing with the \textit{1-median}
%problem; in other words, it can exploit complex networks' analysis
%information for driving us faster to the optimal/suboptimal
%location. At the same time it yields a non-dependable on initial
%location (i.e. service generator) solution and, thus, guarantees
%fairness among various users.

%% file: bc.tex
\section{Weighted Conditional Betweenness Centrality}
\label{sec:metrics}

%\textbf{Main points:} \\
%a) Present the generic (weighted metric), starting from the known
%BC and going through the (unweighted) CBC, and present the heuristics behind it \\
%
%\textbf{Tentative to-do list:} \\
%a) Could elaborate a bit more on the metric -- argue that it
%combines the influence of two things: topology and demand
%distribution \\
%    i) maybe show its value on simple graphs (ring, grid), to make clear how it
%    can assign different importance to network nodes?  \\
%
%NOTE : in that case we would plot the weighted BC (rather than
%CBC, which changes for each ref dst node). Should we present the
%metrics as BC -> wBC -> CWBC? \\
%b) cut text --make metric presentation more compact \\
%
%The solution proposed herein relies heavily upon the criterion
%based on which the subgraph is picked...
%, around the current node having the content, is picked in order to solve (there) the small-scale optimization problem....
% This criterion (referred to as $CBC$ criterion hereafter) involves an innovative, to the best of our knowledge, metric coming from social networks analysis, that captures network traffic towards a specific node (the node, each time, hosting the content).

\emph{Central} to our \emph{distributed} approach is the Weighted
Conditional Betweenness Centrality ($wCBC$) metric. It originates
from the well-known betweenness centrality metric and captures
both topological and service demand information for each node.
%It is used in
%each local execution of our algorithm (see Section
%\ref{sec:analysis}) for two purposes: a) it identifies the subset
%of nodes with the highest potential to host the service, as viewed
%from the current service location; b) it directly maps the demands
%of the remaining nodes on the selected subset and computes the
%effective demand values appearing on the right-hand side of
%Eq.~(\ref{eqn:1-median}). The metric neatly captures both
%topological and user demand information about each node.

% RELEVANT TO DISCUSSION SECTION

%Centrality measures (or indices) are widely used in social network
%analysis... (what are the asumptions made when using metrics
%involving shortest paths? - Borgatti's work on information
%flow~\cite{Borgatti_05})
%
%Centrality measures have been widely used in Complex Networks
%Analysis, as graph-theoretic tools to shed some light on various
%(\eg social) phenomena, since Freeman's late '70s influential
%articles~\cite{citeulike:392801},~\cite{citeulike:278955}. These
%measures, defined either on the nodes or edges of a graph, are
%usually based on geodesic paths that link members of the network
%and aim to provide a notion of importance of one's
%position~\footnote{under the assumption that importance is equally
%divided among all shortest paths of each pair.}. Different
%measures have been introduced to capture a variation of a node's
%importance, like the ability to reach numerous nodes via relative
%short paths or the popularity among
%others~\cite{citeulike:278955}.

\subsection{Capturing network topology: from BC to CBC}
\label{BC}

Betweenness centrality, one of the most frequently used metrics in
CNA, reflects to what extent a node lies on the shortest paths
linking other nodes. Let $\sigma_{st}$ denote the number of
shortest paths between any two nodes $s$ and $t$ in a connected
graph $G=(V,E)$. If $\sigma_{st}(u)$ is the number of shortest
paths passing through the node  $u \in$\textit{V}, then the
\textit{betweenness centrality} index of node \textit{u} is given
by (\ref{eqn:bc}).

\begin{equation}\label{eqn:bc}
%betweenness centrality
BC(u)=\sum_{s=1}^{|V|}\sum_{t=1}^{s-1}
\frac{\sigma_{st}(u)}{\sigma _{st}}
\end{equation}
$BC(u)$ captures the ability of a node $u$ to control or assist
the establishment of paths between pairs of nodes. It is an
average value estimated over all network pairs.
%\subsection{Conditional Betweenness Centrality}\label{CBC}
In earlier work ~\cite{Pantaz2010}, we proposed Conditional BC
(CBC),
%is -to the best of our
%knowledge- Complex Networks Analysis-inspired metric was proposed
%in~\cite{Pantaz2010} for the efficient capturing of the traffic
as a way to capture the topological centrality of a random network
node with respect to a specific node $t$, which in our context is
a node visited by the service on its way towards the optimal
location. It is defined as
%in Eq.(\ref{eqn:cbc})
%Which information, with respect to \ref{uni}, did we try to 'capture'?
\begin{equation}\label{eqn:cbc}
 %conditional betweenness centrality
CBC(u;t)=\sum_{s \in V,u \neq t
}\frac{\sigma_{st}(u)}{\sigma_{st}}
\end{equation}
with $\sigma_{st}(s)=0$.

Note that the summation is over all node pairs ($x,t$) $\forall x
\in V$ destined at node $t$ rather than all possible pairs, as in
(\ref{eqn:bc}). Effectively, CBC assesses to what extent a node
$u$ acts as a \emph{shortest path aggregator} towards the current
service location $t$ by enumerating the shortest paths to $t$
involving $u$ from all other network nodes.
%flow towards a specific node $t$, where some content (or, in our
%case, the service) is currently hosted. This flow -that shapes the
%resulting cost of service provisioning from that node- is the one
%between all node pairs ($x,t$) $\forall x \in V$ for the
%\textit{fixed} node $t$, and not all possible pairs, as expressed
%in (\ref{eq1}). Consequently, if we were to select a subgraph of

%nodes to solve a small-scale optimization problem, it would make
%sense to include the nodes that stand between the ``most'' paths
%linking the network nodes to the specific one hosting the service,
%provided that the all nodes are characterized by equal demand
%load; the presence of such nodes would reflect somewhat the fact
%that relatively high demand (that shapes the resulting cost) is
%coming through such nodes and so they can be used as demand load
%aggregators. The \textit{conditional betweenness centrality} (or
%$CBC$) index defined by eq.(\ref{eq2}), can be used as a criterion
%for constructing the network subgraph (assuming
%$\sigma_{st}(s)=0$):

%The key claim of how to exploit the -in above sense- critical
%position of some nodes in a way that they can be regarded as
%demand load aggregators, obviously needs to be extended to cover
%the non-uniform demand case. The one-to-one relation between the
%number (or density) of nodes lying in some area and the emerging
%demand load, from that very area, does not hold; in general, a
%great number of shortest paths passing-through (the node $u$) does
%not necessarily mean that equally great demand load stems from the
%``source-area'' of those paths.

\subsection{Capturing service demand: from CBC to wCBC}
% As earlier explained the demand load is a major factor in identifying the content's optimal location.,
In general, a high number of shortest paths through the node $u$
does not necessarily mean that equally high demand load stems from
the sources of those paths. Naturally, we need to enhance the pure
topology-aware $CBC$ metric in a way that it
%would no more be oblivious about%
takes into account the service demand that will be eventually
served by the shortest paths routes towards the service location.
To this end, we introduce \textit{weighted conditional betweenness
centrality} ($wCBC$), where the shortest path ratio of
$\sigma_{st}(u)$ to $\sigma _{st}$
%summed over each node $s$ other than $t$ as described%
in Eq.~(\ref{eqn:cbc}), is modulated by the demand load generated
by each node $s$.
\begin{equation}\label{eqn:wCBC}
%betweenness centrality
wCBC(u;t)=\sum_{s \in V,u \neq t } w(s) \cdot
\frac{\sigma_{st}(u)}{\sigma _{st}}.
%C(u)=\sum_{s=1}^{|V|}\sum_{t=1}^{s-1} \frac{\sigma_{st}(u)}{\sigma _{st}}
\end{equation}
Note that $ \sigma_{ut}(u) = \sigma _{ut}$; hence, for each
network node $u$, its $wCBC(u;t)$ value is lower bounded by its
own demand $w(u)$.
%where the equality $ \sigma_{ut}(u) = \sigma _{ut}$
%results in
%assigning an equal to $u$'s demand $wCBC$ value (otherwise equal
%to zero), not only to those being intermediate between $s$ and
%$t$, but also to any one-hop neighbors of the host node $t$.

Therefore, $wCBC$ assesses to what extent a node can serve as
\emph{demand load concentrator} towards a given service location.
It is straightforward to see that when a service is equally
requested by all nodes in the network (uniform demand) the $wCBC$
metric degenerates to the CBC one, within a scale constant.

\subsection{Metric computation for regular network topologies}

Closed-form expressions for $wCBC$ are not easy to obtain except
for scenarios with uniform demand and regular topologies. The
following two Propositions provide the closed-form expressions for
CBC, \ie $wCBC$ for $w_n=1,~\forall n \in V$, in two instances of
regular network topologies, the ring and the two-dimensional (2D)
grid.

\begin{proposition}\label{prop_ring}
In a ring network of $N$ nodes, the CBC value of a node $u$ with
respect to another node $t$ are given by:
%\begin{equation} \label{eqn:cbc_ring}
 \[ CBC_{ring(N)}(u;t) = \begin{cases}
\lceil\frac{N-1}{2}-d(u,t))\rceil^+~~~N = 2k \\
 \lceil\frac{N+1}{2}-d(u,t))\rceil^+~~~N=2k+1, ~k \in Z
\end{cases}\]
%\end{equation}
where $\lceil x\rceil^+ = max(x,0)$ and $d(u,t)$ is the minimum
hop count distance between nodes $u$ and $t$ along the ring.
\end{proposition}

\begin{proposition}\label{prop_grid}
Consider a $M$x$N$ rectangular grid network, where nodes are
indexed inline with their position in the grid, \ie node $(i,j)$
is the node located at the $i^{th}$ row and $j^{th}$ column of the
grid. The CBC value of node $u$ at position $(a,b)$ with respect
to node $t$ at position $(k,l)$ is given by (\ref{eqn:cbc_grid}).
%\begin{multicols}{1}
%\footnotesize

\begin{figure*}[ht]
%\label{wons}
\begin{minipage}[b]{.55\linewidth}
\centering
\framebox(205,180){\includegraphics[width=6.3cm]{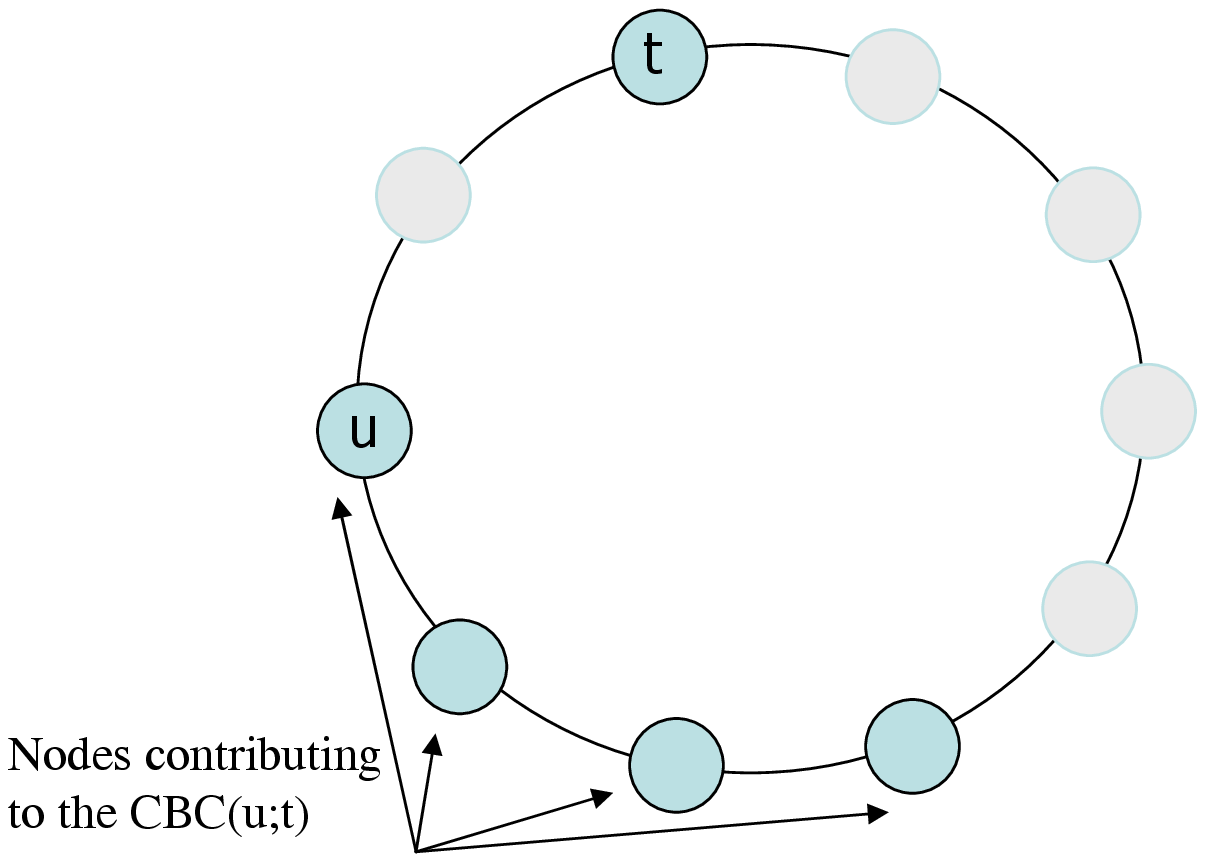}}
%\caption[short]{}
\end{minipage}
\hspace{0.35cm}
\begin{minipage}[b]{0.4\linewidth} \centering
\framebox(205,180){\includegraphics[width=4.5cm, bb = 23 95 358 416]{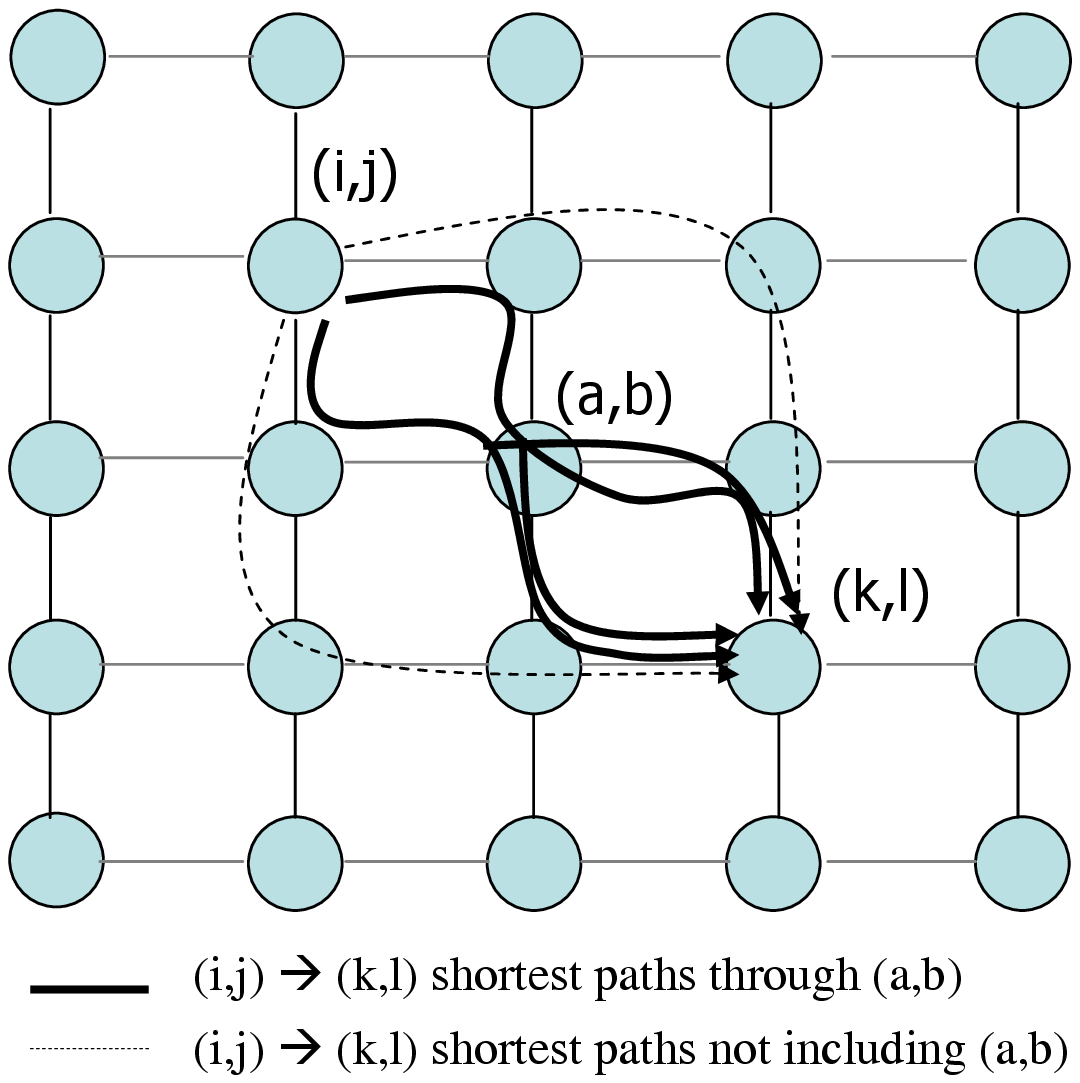}}
%\caption[short]{}
\end{minipage}
\caption{Conditional Betweenness Centrality in regular
topologies.} \label{cbc_plots}
\end{figure*}

\begin{figure*}[tb]
\begin{equation} \label{eqn:cbc_grid}
CBC_{grid(M,N)}(u;t) = \sum_{i=1}^{M} \sum_{j=1}^{N}
\frac{\binom{|b-j|+|a-i|}{|a-i|}
\binom{|l-b|+|k-a|}{|k-a|}}{\binom{|l-j|+|k-i|}{|k-i|}}1{\hskip
-2.5 pt}\hbox{I}_{\{|l-j|+|k-i| == |b-j|+|a-i|+|l-b|+|k-a|\}}
\end{equation}
%\footnotesize
\end{figure*}
\end{proposition}
\begin{proof}
The proof of the first proposition is straightforward. There is
one minimum hop count path between all pairs of nodes in the ring.
The only exception concerns nodes $N/2$ positions away the one
from another in rings with even number of nodes, where there are
two shortest paths. For given destination node $t$, the $CBC(u,t)$
value is only increased by those shortest paths that encompass the
intermediate node $u$. Due to the ring symmetry, their number only
depends on the distances between nodes $u$ and $t$ and decreases
by one for each additional hop away from $t$. Summing them over
the respective half of the ring, yields the result.

For the 2D grid, the problem degenerates into the enumeration of
shortest paths between two grid nodes. The denominator of
(\ref{eqn:cbc_grid}) expresses the number of shortest paths
between two arbitrary nodes $(row, column)$ coordinates $(i,j)$
and $(k,l)$, whereas the  numerator of (\ref{eqn:cbc_grid}) equals
the number of those paths going through a node with coordinates
$(a,b)$. We then sum the ratios over all grid nodes with shortest
paths to node $t=(k,l)$ encompassing node $u=(a,b)$.
\end{proof}

%To include in here the ``dig more into metric'' tasks 1 \& 2 of
%the working plan, i.e. upper bounded metric, monotonically
%increasing on single shortest path route to target, correlation of
%CBC values in known graphs(?), any other? What is the
%configuration of the top CBC nodes around the target? Is there any
%locality attribute- are we really interested in it?

%% file: analysis.tex
\section{The cDSM Algorithm description}
\label{sec:analysis}
%In this section we present the CNA-driven
%approach to the content placement problem. Our solution's
%pseudocode is depicted in the Algorithm~\ref{alg1}, where the cost
%of placing the content at node $k$ of graph $G(V,E)$ that
%represents the entire topology, is denoted by $C(k)$. The  $i$-th
%iteration's subgraph around the host-node $n$, where the
%small-scale optimization takes place, is denoted by $G_{n}^{i}$.
%Finally, we use the variables $C_{current}$ and $C_{next}$ to hold
%the cost of current and next step solution, respectively.

In this section we present our CNA-driven Disributed Service Migration Algorithm for the
service placement process. The service migration to the optimal
location in the network evolves within a finite number of
iterations, as we show later in Section \ref{subsec:convergence},
through a path that continuously decreases the aggregate cost of
service access over all network nodes.

%Our solution's pseudocode is depicted in the Algorithm~\ref{alg1},
%where the cost of placing the content at node $k$ of graph
%$G(V,E)$ that represents the entire topology, is denoted by
%$C(k)$. The  $i$-th iteration's subgraph around the host-node $n$,
%where the small-scale optimization takes place, is denoted by
%$G_{n}^{i}$. Finally, we use the variables $C_{current}$ and
%$C_{next}$ to hold the cost of current and next step solution,
%respectively.

\subsection{Detailed algorithm description}
A single algorithm iteration involves a number of discrete steps.
We discuss them below while providing pointers to the algorithm's pseudocode.

\emph{Step 1: Initialization}. The algorithm execution starts at
the node $s$ that initially deploys the service. The cost of the
service placement at node $s$ is assigned an infinite value
$(line~3)$ to secure the first algorithm iteration $(line~11)$.
This step is only relevant to the first algorithm iteration.
%In subsequent iterations, the new reference node is the node
%currently hosting the service.

\emph{Step 2: Metric computation and 1-median subgraph
derivation}. Next, the computation\footnote{For our simulation's
needs, this involves solving the all-pairs shortest path problem.
Common algorithms, like Floyd-Warshall~\cite{citeulike:201727},
may need even $\Theta{(|V|^{3})}$ time to solve, on a $G(V,E)$
graph. Hence, for $weigthedCBC$ computation we properly modified a
scalable algorithm~\cite{Brandes01afaster} for \textit{betweenness
centrality}, with runtime $\mathcal{O}(|V||E|)$. The cost
introduced is low, as the length and number of all shortest paths
from a given source to every other node, needed for our
computation, is determined in $\mathcal{O}$\textit{(|E|)}
time~\cite{Brandes01afaster}.} of $wCBC(u;s)$ metric takes place
for every node $u$ in the network graph $G(V,E)$. Nodes featuring
the top $\alpha\%~wCBC$ values together with the node $Host$
currently hosting the service form the subgraph $G_{Host}^{i}$
($i$ enumerates the algorithm iterations), over which the 1-median
problem will be solved $(lines~4-5~and~14-15)$. Clearly, the size
of this subgraph and the complexity in the problem solution are
directly affected by the choice of the parameter $\alpha$. We show
in Section \ref{sec:evaluation}, that even with very small
$\alpha$ values, our algorithm yields solutions very close to the
optimal.

\emph{Step 3: Mapping the demand of the remaining nodes on the
subgraph}. In this step, the service demand from nodes in
$G\setminus G_{Host}^{i}$ is mapped to the nodes of the
$G_{Host}^{i}$ subgraph that explicitly participate in the
$1$-median problem solution. How this is done is described in
detail in section \ref{subsec:mapping}. For the moment, it
suffices to say that the demand factors $w(n)$ in Eq.
(\ref{eqn:1-median}) are \emph{effective demands},
$w_{eff}(n;Host)$, dependent on the current service host. They
include not only the demands of the nodes selected in the previous
step due to their high $wCBC$ values but also the demands of the
remaining nodes that are not directly considered in the 1-median
problem formulation.

\emph{Step 4: 1-median problem solution and service migration to
the new host node}. Any centralized technique may be used to solve
this small-scale optimization problem. Successively better
algorithms have been designed during the last few
years~\cite{alg_k_Med} and one can seek for the best heuristic
method available to maximize scalability.
%\textbf{For simplicity
%reasons we have solved it using
%enumeration~\cite{mirchandani1990}, even if there are other more
%cost-effective methods. Besides, this implementation choice
%affects neither the degree nor the speed of convergence of our
%solution to the optimal one.}
The optimization's outcome is the
location of the candidate new \textit{Host} node, which results in
minimum service access cost $C(Host)$ \footnote{In case of
multiple minimum-cost solutions within the $G^{i}$ nodes, we
choose randomly one of them.} among the nodes of the current
subgraph. We assign the value of this cost to the variable
$C_{next}$ and test whether it is smaller than $C_{current}$.
%In case it is true, the
%content is moved to node $Host$ and steps 1-3 are repeated. 
%\footnote{The percentage $\alpha$ of total number of nodes, that
%forms the $G^{i}$ subgraph, remains the same until our algorithm's
%completion.}.
As long as the condition for cost decrease holds, the service is
being relocated to this node, the algorithm iterates again through
steps 2-4, and the service continues its progress towards the
(globally) lowest-cost location.
%In the next paragraph, we prove
%that this service migration path is finite.

% 
\begin{algorithm}[ht]
\caption{cDSMA in \textit{G(V,E)}}
\label{alg1}
\begin{algorithmic}[1]
\STATE $choose \ randomly \ node \ s $ 
\STATE $place \ SERVICE\ @ \ s\ $ 
\STATE $C_{current} \leftarrow \infty $
\STATE $  \textbf{for all} \ u \in G \  \textbf{do} \ compute \  wCBC(u;s)$ 
\STATE $G_{s}^{o} \leftarrow \{\alpha \% \ of \ G \ with \ top \ wCBC\ values\} \cup \{ s\}$ 
\STATE $ \textbf{for all} \ u \in G_{s}^{o} \ \textbf{do} $
\STATE $ \ \ \  compute \ w_{map}(u;s)$
\STATE $  \ \ \ w_{eff}(u;s) \leftarrow w_{map}(u;s) +w(u)$ 
\STATE $ Host \ \leftarrow  \textit{1-median} \ solution \ in \ G_{s}^{o}   $ 
\STATE $C_{next} \leftarrow C(Host), ~ i \leftarrow 1 $
\WHILE{$C_{next} < C_{current}$} 
\STATE $move \ SERVICE \ to \ Host $ 
\STATE $C_{current} \leftarrow C_{next} $ 
\STATE $ \textbf{for all } \ u \in G \  \textbf{do} \ compute \ wCBC(u;Host) $ 
\STATE $G_{Host}^{i} \leftarrow \{\alpha \% \ of \ G \ with \ top \ wCBC\ values\} \cup \{ Host\} $ 
\STATE $ \textbf{for all} \ u \in G_{Host}^{i} \  \textbf{do}  $ 
\STATE $ \ \ \ compute \ w_{map}(u;Host)$
\STATE $ \ \ \  w_{eff}(u;Host) \leftarrow w_{map}(u;Host) +w(u)$ 
\STATE $NewHost   \leftarrow \textit{1-median} \ solution \ in \ G_{Host}^{i} \  $
\STATE $Host \leftarrow  NewHost $ 
\STATE $C_{next} \leftarrow C(NewHost),~i \leftarrow i+1 $ 
\ENDWHILE
\end{algorithmic}
\end{algorithm}

\subsection{On the convergence of the proposed
algorithm}\label{subsec:convergence}

In this paragraph we study the convergence of cDSMA, showing that the migration
process terminates after a finite number of steps. The following lemma serves
as the basis for the proof of the convergence proposition. 
% In this paragraph we demonstrate that the cost-decrease criterion
% is a sufficient condition for the termination of the algorithm in
% a finite number of steps.
%We sketch a proof for its convergence below:
\newtheorem{theorem3}{Lemma}%[section]
\begin{theorem3}
\label{lemma_convergence}
A service facility following the migration process of Algorithm~\ref{alg1} 
will visit at most one network node twice. 
\end{theorem3}
\begin{proof}
Assume that the service, initially deployed at some node $n \in G$ 
reaches the node $b \in G$ twice. Right after its first placement
at $b$ upon iteration, say, $i-1$ we solve the 1-median
in the subgraph $G_{b}^{i}$ that is formed by the nodes with the
top $\alpha \%~wCBC(u;b)$ values. Let the corresponding cost be
$C_{b}^{i}$. When the service returns to $b$ at iteration, say,
$j$ given that the network topology remains the same, the
deterministic $wCBC$ criterion of (\ref{eqn:wCBC}) singles out the
same subgraph with the one of the first visit, so we have that
$G_{b}^{i} = G_{b}^{j}$, implying for the costs that
$C_{b}^{i}= C_{b}^{j}$; the cost-decreasing
condition of cDSMA is then not fulfilled and, thus, 
the service locks at node $n$ and the migration process halts.
\end{proof}

\begin{proposition}\label{prop_convergence}
cDSMA converges at some solution in $O(|V|)$ steps.
\end{proposition}
\begin{proof}
As stated above, the solution derived from cDSMA is either
the globally optimal (best case) or one locally anticipated as
lowest-cost solution. Since the number of network nodes is finite, 
the migrating service will -according to \textit{Lemma~\ref{lemma_convergence}}- visit at most
every node once and only one of them, twice. This takes $O(|V|+1)=O(|V|)$ steps.    
\end{proof}

\subsection{Mapping the demand of remaining nodes}\label{subsec:mapping}
%The $wCBC$ metric naturally captures how the demand load for the
%service is routed/concentrated through the network topology all
%the way towards the service host. Therefore, it lets us select the
%subset of the ``most significant'', in this respect, nodes for
%solving the \textit{1-median} optimization problem within the
%induced 1-median subgraphs $G_{Host}^{i}$.

Besides being the basis for extracting the 1-median subgraph
$G_{Host}^{i}$ in each algorithm iteration, the $wCBC$ metric also
eases the mapping of the demand that the rest of the network nodes
in $G\backslash G_{Host}^{i}$ induce on the 1-median subgraph.
This demand must be taken into account when solving the 1-median
problem. We do this by modulating the original $wCBC$ metric in
accordance with two observations.

%A final step, before solving the optimization problem in the
%subgraph $G^{*(i)}$ of the top a\% \textit{wCBC} valued nodes for
%the $i$-th iteration, has to do with correctly
%mapping\footnote{Our simulation studies indicated that some extra
%piece of information is needed in order to somehow discriminate
%between the amount of demand that two of the already selected -by
%the proposed metric- nodes may possibly receive. This information
%is effectively provided by the mapping mechanism.} the demand of
%the remaining nodes on them. To this end, we employ the idea of
%tagging each selected node $v$ ($v \in G^{*(i)}$) with a weight
%that would reflect what portion of the demand generated by the
%``world outside" the subgraph (i.e., the non-shaded portion of $G$
%in Fig.~\ref{schematic}), passes through $v$.

Firstly, during the computation of the node $wCBC$ values, the
demand of a node $z$ in $G\backslash G_{Host}^{i}$ is taken into
account in all the $G_{Host}^{i}$ nodes that lie on the shortest
path(s) of $z$ towards the service host node $t$. Simply mapping
the demand of $z$ on all those nodes inline with the original
$wCBC$ metric, has two shortcomings: (a) when the demand of
heavy-hitter nodes is distributed among multiple nodes, any strong
direction(gradient) of heavy demand that would otherwise ``pull''
the service towards a certain direction, tends to fade out; (b)
the cumulative demand that is mapped on all $G_{Host}^{i}$ nodes
ends up exceeding considerably the real demand a node poses for
the service. For example, in Fig.~\ref{fig:schematic} let w(16)=
$\delta$; then naive reuse of the $wCBC$ values for service demand
mapping would result in nodes $11$, $8$ and $12$ receiving 100\%,
50\% and 50\% of the original $\delta$ demand, respectively.
%Thus, not only is the demand load of $s$ being mapped on many of
%the $G^{*(i)}$ nodes (smoothing out, throughout the $G^{*(i)}$
%nodes, any in-coming path ``carrying'' heavy demand), but also it
%results in the $G^{*(i)}$ nodes being cummulatively assessed
%greater demand than the one actually generated by $s$
%\footnote{Let -for instance- w(16)=$\delta$ in
%Fig.~\ref{schematic}; then the $wCBC$-based mapping would result
%in nodes 11, 8 and 14 receiving 100\%, 50\% and 50\% of the
%original $\delta$ demand, respectively.}.
Hence, to achieve accurate mapping, the influence of $z$ should be
``credited'' only to the first $G_{Host}^{i}$ node encountered on
each shortest path from $z$ towards the service host. The set of
all these \emph{entry} nodes $v$ with this property forms a
subgraph of $G_{Host}^{i}$.
% denoted as $G_{Host,b}^{i}$.%
%and lying
%on the $j$-th shorthest path between $z$ and the service host at
%the location where the distance from $z$, $d^{j}_{zt} (z,v)$, is
%minimum.

Secondly, it happens frequently that the shortest paths
originating from the 1-median subgraph nodes include further
subgraph nodes. The demand of those nodes have to be subtracted
when computing the effective demand, with which each
$G_{Host}^{i}$ node participates in the solution of the 1-median
problem since they are accounted for directly through the very
same nodes that generates them.
%by the \textit{wCBC} values estimated originally for
%those nodes since they are already included in the optimization
%problem through the very same nodes that generate them.

Mathematically speaking, the weights $w(n)$
%with which the nodes
%in $G_{Host}^{i}$ participate in the solution of the 1-median
%problem
in Eq. (\ref{eqn:1-median}) can be regarded as \emph{effective
demands}
\begin{equation}\label{eqn:demand}
\mbox{\fontsize{10}{9}\selectfont $w_{eff}(n;Host)=
 w(n)+w_{map}(n;Host)$}
%\footnotestyle{New\,Demand(u;Host)= Selected\,CBC(u;Host)+initial\,demand(u)}
\end{equation}
%with which each node $v \in G^{*(i)}$ will take part in the
%optimization problem.
that bring together two terms. The first one is the native demand
for the service coming from users that are served by node $n$. The
second term corresponds to the contribution of the nodes in
$G\backslash G_{Host}^{i}$ (\ie the non-shaded nodes in
Fig.~\ref{fig:schematic}), which is given by:

\begin{eqnarray}\label{eqn:selCBC}
%betweenness centrality
w_{map}(n;t) & = &  \sum_{s \in\{ G\setminus G_{Host}^{i} \} }w(s)
\frac{\sigma'_{st}(n)}{\sigma_{st}} \\
\sigma'_{st}(n) & = &  \sum_{j=1}^{\sigma_{st}} 1{\hskip -2.5
pt}\hbox{I}_{ \{n~\in~SP_{st}(j) \bigcap n = \underset{u \in
SP_{st}(j)}{\operatorname{argmin}} d(s,u) \} } \nonumber
%argmin_{n \in SP_{st}(j) } d(s,n) \} }
% \sum_{j\in [1,\sigma_{st}]}\frac{\sigma_{st}(v|min_{d^{j}_{st}
% (s,v)})}{\sigma _{st}}
%C(u)=\sum_{s=1}^{|V|}\sum_{t=1}^{s-1} \frac{\sigma_{st}(u)}{\sigma _{st}}
\end{eqnarray}
where $SP_{st}(j)$ is the set of shortest paths from node $s$ to
node $t$.

\begin{figure}[ht]
  \centering
%\framebox(180,135){\includegraphics[scale=0.2]{draw3}}
\framebox(190,170){\includegraphics[scale=0.22]{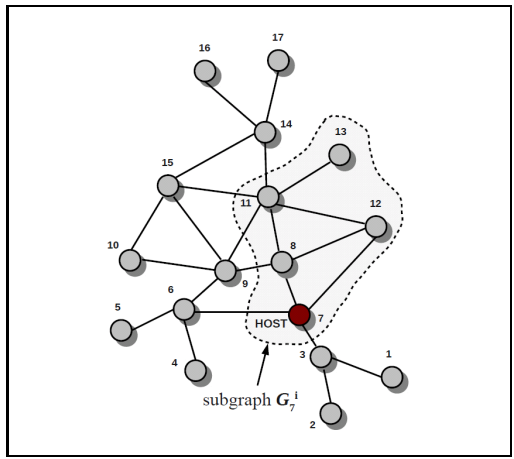}}
 % \hspace{0.5cm}
  \caption{The 1-median subgraph for an example network of 17 nodes with node 7 storing the service facilities in the $i^{th}$ algorithm iteration.
   There are two non-zero demand mapping terms, $w_{map}(8;7)$ and $w_{map}(11;7)$.}
%  \label{schematic}
\label{fig:schematic}
\end{figure}

Back to our example in Fig. \ref{fig:schematic}, the nodes $14$,
$15$, $16$ and $17$ will now contribute to the $w_{map}(11;7)$
value, whereas the included in $G_{7}^{i}$, node $13$, will not.

%% file: evaluation.tex
\section{Evaluation methodology}\label{sec:evaluation}
%We have implemented the algorithm described above and carried out
%a number of simulations both on well-studied synthetic graph
%models and realistic (intra-ISP) topologies to evaluate its
%performance.
It should have become clear by this point that both the $wCBC$
metric and the performance of cDSMA are heavily dependent
on two factors: the network topology and the service demand
distribution within the network. Their combination may enforce or,
on the contrary, suppress strong service demand attractors and
assist (resp. impede) the progress of the service facilities
towards their optimal location in the network. In what follows, we
study the behavior of cDSMA over a broad set of scenarios
that cover efficiently the $\{net~topology, demand~distribution\}$
variation space.

%\subsection{Simulation methodology and settings}
% Set-up description: which synthetic graphs do we use? Number of
% nodes, parameters, etc...
% Why do we employ (those) random graphs? Do they reflect any
% real-world scenario?

\subsubsection{Network topology} We consider both
synthetic and realistic network topologies. The two synthetic
topologies we experiment with are the
Barab{\'{a}}si-Albert~\cite{RefWorks:2503} and two-dimensional
rectangular grid graphs. The two types of graph models bear very
different and distinct structural properties. The B-A graphs form
pure probabilistically and can reproduce a highly skewed node
degree distribution that approximates the power-law shape reported
in literature~\cite{falouts}. Grids, on the other hand, exhibit
strictly regular structure with constant node degree and diameter
that grows exponentially with the number of network nodes.
%illustrating topology's role in identifying the optimal location.
%The E-R graph model and the Grid topology can be considered as the
%two marginal cases as far as the topological structure is
%concerned
%The former results in forming a pure probabilistically generated
%network, whereas the latter is the typical case of a
%``strictly'' regular topology.
The synthetic network topologies let us assess the algorithm and
highlight its behavior under certain extreme yet predictable
operational conditions. Nevertheless, the ultimate assessment of
our algorithm is carried out over real-world ISP network
topologies. The dataset we consider ~\cite{dataset} has been
recently made publicly available~\cite{ISP_data,pam10}. It
includes topology data from 850 distinct snapshots of 14 different
AS topologies, corresponding to five Tier-1, five Transit and four
Stub ISPs. The data were collected daily during the period
2004-08 with the help of a multicast discovering tool called
\emph{mrinfo}. The tool uses IGMP~\cite{RFC1112} messages to
recursively probe all IPv4 multicast-enabled routers and receive
back all their multicast interfaces as well as the IP addresses of
their neighboring routers. At a second step, the borders between
ASes are delimited with application of two mapping mechanisms:
firstly, an IP-to-AS mapping for assigning a number to each AS
and, secondly, a router-to-AS mapping, via both probabilistic and
empirical rules, for assigning each router (having multiple IP
addresses) to the ``correct'' AS. The method can discover
connections through L2 switches and turns out to be providing an
accurate view of the network topology, circumventing the
complexity and inaccuracy of more conventional measurement tools
such as traceroute.

%Our experiments were first
%conducted on 3 synthetic graph models of N nodes, namely the
%Erd{\H{o}}s-R{\'{e}}nyi~\cite{Erdos1959},
%Barab{\'{a}}si-Albert~\cite{RefWorks:2503} and Grid network, each
%one exhibiting diverse structural properties and, thus,
%illustrating topology's role in identifying the optimal location.
%The E-R graph model and the Grid topology can be considered as the
%two marginal cases as far as the topological structure is
%concerned. The former results in forming a pure probabilistically
%generated network, whereas the latter is the typical case of a
%``strictly'' regural topology. Finally, the B-A model is employed
%herein as it can reproduce a more realistic -than the former two-
%network topology (the degree distribution of the Internet graph is
%known to satisfy a power law~\cite{falouts}). To further explore
%the applicability of our heuristic on real-world networks as well
%as support the introduction's case study, we have performed
%simulations, depicted in~\ref{subsec:data_runs}, using a recently
%collected intra-ISP dataset of physical topologies\cite{ISP_data}.

\subsubsection{Service demand distribution} Our
assessment, at first level, distinguishes between uniform and
non-uniform demand scenarios. Uniform demand scenarios are far
from realistic; yet they let us study the \emph{exclusive} impact
of network topology upon the behavior of the algorithm. On the
contrary, under non-uniform demand distributions, we assess the
algorithm under the \emph{simultaneous} influence of network
topology \emph{and} service demand dynamics. Mathematically
speaking, a Zipf distribution models the preference $w(n;s,N)$ of
nodes $n, n \in \mathcal{N}$ to a given service

\begin{equation}
   w(n;s,N) = \frac{1/n^{s}}{\sum_{l=1}^{N}{1/l^{s}}}.
\end{equation}

Practically, the distribution could correspond to the normalized
request rate for a given service by each network node. Increasing
the value of the parameter $s$ from $0$ to $\infty$, the
distribution asymmetry grows from zero (uniform demand) towards
higher values.

At a second level, we consider two options as to how the
non-uniform service demand emerges spatially within the network.
In the default option, each node randomly generates demand
according to the Zipf law. The alternative is to introduce
geographical correlation by concentrating nodes with high demand
in the same network area. This second scenario lends itself to
modelling services with strongly local scope; whereas, the first
one matches better services that attract geographically broader
interest.

%We conducted (and present) our experimental
%study according to the two main patterns involving service demand;
%the uniform and non-uniform one. With respect to our problem's
%-two- key factors, it is expected that the topology will assert
%the dominant role in the former case while both the topology and
%demand will shape the provision cost (as well as the introduced
%metrics) in the latter. At a second-level, under non-uniform
%demand, we have investigated the case of geographically correlated
%demand generation, practically reflecting the social dimension of
%having co-located users being interested in the very same service.
%-------------------------------
% to clearly illustrate the effect of spatial demand correlation as well as draw safe conclusion,
% we somehow try to eliminate the topology factor by using a grid topology for our experiments.
% But you    % dont need to mention that here, since you just present the evaluation methodology at a high revel.
%  Right ??????
%--------------------------------

\subsubsection{Algorithm performance metrics} We are concerned with two
metrics when assessing the performance of cDSMA. The first
one relates to its accuracy and denotes the degree of convergence
of our heuristic solution to the optimal one, as derived by using
ideal global topology and demand information. It is defined as the
\emph{average normalized excess cost}, $\beta_{alg}$, and equals
the ratio of the service access cost our algorithm achieves,
$C_{alg}(G,\overline{w})$, over the cost achieved with the optimal
solution, $C_{opt}(G,\overline{w})$, for given network topology
$G$ and service demand distribution $\overline{w}$:
\begin{equation}\label{eq9a}
 \beta_{alg}(\alpha;G,\overline{w}) = E[\ \frac{C_{alg}(\alpha;G,\overline{w})}{C_{opt}(\alpha;G,\overline{w})} \ ]
\end{equation}

Clearly, $\beta_{alg}$ depends on the percentage $\alpha$ of the
network nodes participating in the solution. Generally, the error
induced by our heuristic decreases for larger $1$-median
subgraphs, \ie greater $\alpha$ values. Closely related to
$\beta_{alg}$ and its variation are the indices
$\alpha_{\epsilon}$, corresponding to the minimum values of
$\alpha$, where the access cost achieved with our heuristic
algorithm falls within $100\cdot \epsilon \%$ of the optimal.
\begin{equation}\label{eq9b}
\alpha_{\epsilon} = argmin~\{\alpha | \beta_{alg}(\alpha) \leq\
(1+\epsilon)\}
\end{equation}

The second metric is the \emph{migration hop count}, $h_m$, which
is generally a function of the percentage $\alpha$ and reflects
how fast the algorithm converges to its (sub)optimal solution--the
question of whether it does so has been answered in Section
\ref{sec:analysis}. Smaller $h_m$ values imply faster service
deployment and less overhead involved to transport and service
set-up/shut-down tasks.

For any chosen configuration of the involved parameters, we repeat
20 simulation runs to achieve statistical significance.
Typically, the results plotted hereafter are average values
together with the 95\% confidence intervals, estimated over the
20 runs.
%\footnote{ An
%estimated range of values which is likely to include the mean is
%given by a 95\% confidence interval, associated with each sample
%of runs.}
%over the runs.

%% file: results.tex
\section{Simulation results}\label{sec:results}

\subsection{Synthetic topologies: experiments under uniform demand}\label{subsec:uniform}
%Following the above methodology we evaluate the performance of our
%heuristic under a uniform demand pattern. Although this case may
%not be considered as a realistic one (especially in the context of
%a \textit{User-Generated Service}), it is included in our study as
%the first evaluation step and definitely as a benchmark.

As already explained in Section \ref{sec:evaluation}, these
experiments demonstrate how different topologies may facilitate or
encumber our algorithm. All nodes posing the same demand, the
optimal service location coincides with the node featuring the
%lowest aggregate sum of minimum hop counts to all other network
%nodes, or
minimum reciprocal of closeness centrality~\cite{freeman79}.

\begin{figure}[tbp]
%\label{wons}
\begin{minipage}[b]{.55\linewidth}
\centering
\includegraphics[width=6.cm]{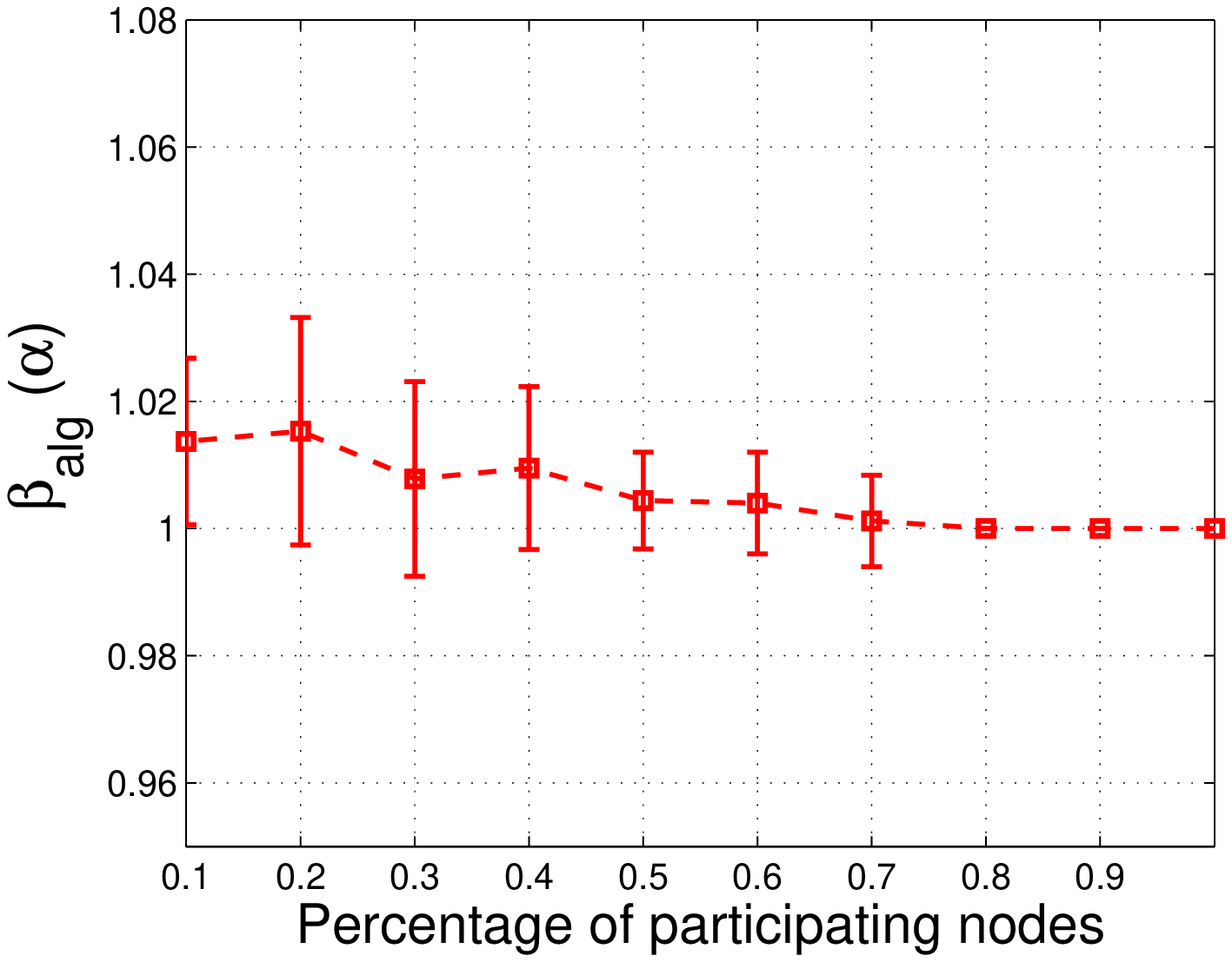}
%\caption[short]{}
\end{minipage}
\hspace{0.35cm}
\begin{minipage}[b]{0.4\linewidth}
\centering
\includegraphics[width=6.cm]{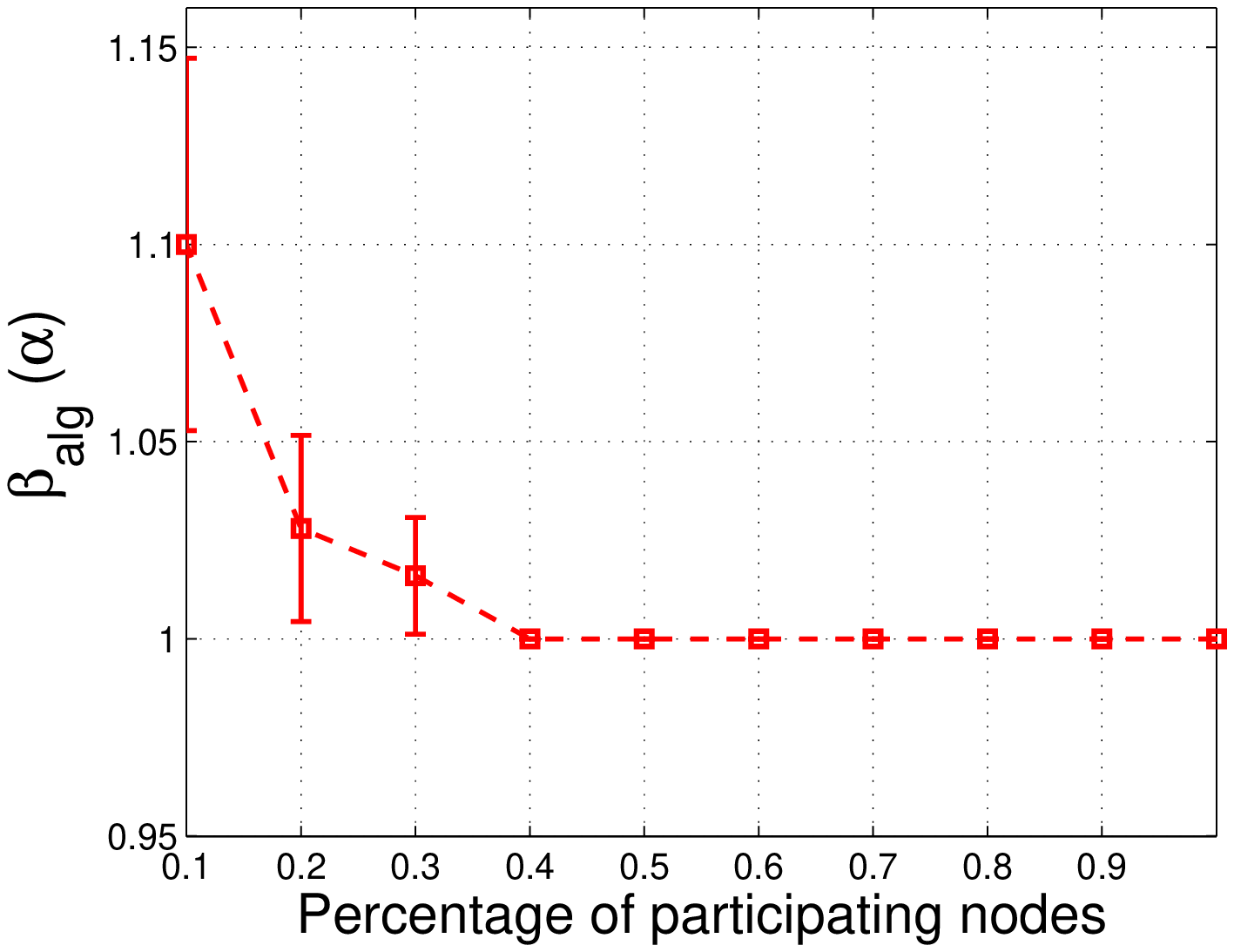}
%\caption[short]{}
\end{minipage}
\caption{cDSMA accuracy for 10x10 nodes B-A graphs(left) and grids(right) as a function
of the 1-median subgraph size: uniform service demand
distribution.} \label{uni_plots}
\end{figure}

% Accuracy
Figure~\ref{uni_plots} plots the average normalized excess cost
$\beta_{alg}$ for B-A %Barab{\'{a}}si-Albert ~\cite{RefWorks:2503}
and grid graphs of 100 nodes. Qualitatively, the two plots are
similar: the error induced by our heuristic decreases
monotonically with
%the percentage of network nodes, $\alpha$,
%participating in
the 1-median subgraph size. However, both starting values,
$\beta_{alg}(0.1)$, and the required subgraph size for achieving
optimal performance, $\alpha_0$, differ. The behavior of cDSMA
on the B-A graph is better. The aggregate service access
cost increase is within $2\%$ of the optimal, even when we include
$10\%$ of network nodes in the 1-median problem solution. On the
contrary, reaching similar accuracy for the grid would require, on
average, no less that $40\%$ of the network nodes.

%The second remark that can be made is that the algorithm performs
%clearly better in B-A than in E-R and grid graphs.
Both grids and B-A graphs have structured connectivity.
Nevertheless, the existence of high-degree nodes, called hubs,
%that dominate connectivity
in B-A graphs, appears to ease more the algorithm operation.
Placing the service on, or nearby, hub nodes suffices for getting
a very good, even when suboptimal, solution, already for small
1-median subgraphs. On the contrary, grids exhibit more regular
structure; all nodes have the same degree and there is smaller
variance in the connectivity properties of neighboring nodes.
Analyzing our simulation runs, we found that the content migration
jumps within the grid are clearly shorter than in B-A graphs; in
many cases the service migrates to neighboring nodes. Even worse,
cDSMA gets more often trapped and terminates prematurely
in suboptimal locations. Said in different way, the attraction
force of the optimal location, \ie the grid center node for odd
$M$ and $N$, a neighborhood around the center otherwise, is not
impelling enough to pull the migrating service all the way to it
except for large enough 1-median subgraphs.
%Its strength, as measured by its $CBC$ value, is rather attenuated
%with the demand mapping process, which .

\begin{table*}[tbp]
\centering \caption{Average normalized excess cost and hopcount for B-A and
Grid networks: uniform service demand} \label{uni_hops}
{\scriptsize
\begin{tabular}{ c || c c c c | c c c c}
\multicolumn{1}{c}{}&\multicolumn{4}{c}{\scriptsize{\textbf{B-A graph}}}&\multicolumn{4}{c}{\scriptsize{\textbf{Grid network}}}\\
\cline{1-9}
& & & & & & & & \\
%\hline
\footnotesize{Network size N} &\footnotesize{$\beta_{alg}(0.1)$}&\footnotesize{ $h_m(0.1)$} &\footnotesize{$\beta_{alg}(0.4)$}&\footnotesize{ $h_m(0.4)$} &\footnotesize{$\beta_{alg}(0.1)$}&\footnotesize{ $h_m(0.1)$} &\footnotesize{$\beta_{alg}(0.4)$}&\footnotesize{ $h_m(0.4)$   }\\
\hline
& & & & & & & & \\
50 \tiny{(25x2)}&1.0453$\pm$0.0524&2.25$\pm$0.31&1.0125$\pm$0.0186&1.95$\pm$0.28&1.0074$\pm$0.0071& 1.40$\pm$0.35 &1.0086$\pm$0.0058&1.10$\pm$0.22\\
%\hline
%& & & & & & & & \\
100 \tiny{(25x4)}&1.0134$\pm$0.0169&2.00$\pm$0.32&1.0070$\pm$0.0164&2.00$\pm$0.00&   1.0569$\pm$0.0333&   1.30$\pm$0.33 & 1.0006$\pm$0.0012 &  1.20$\pm$0.29  \\
%& & & & & & & & \\
%\hline == ==
200 \tiny{(40x5)}&1.0216$\pm$0.0327&2.00$\pm$0.00&1.0028$\pm$0.0061&1.95$\pm$0.16 &   1.0636$\pm$0.0487&1.60$\pm$0.71  &1.0013$\pm$0.0043&2.05$\pm$0.59\\
%& & & & & & & & \\
%\hline
300& 1.0125$\pm$0.0147  &2.00$\pm$0.00& 1.0032$\pm$0.0070 &2.00$\pm$0.00&   & & & \\
& & & & & & & & \\
\hline
\end{tabular}}
\end{table*}

% migration hop
This differentiation in the behavior of cDSMA, hence its
performance, over the two graphs is amplified when we let the
network size and diameter grow. Table \ref{uni_hops} lists the
accuracy and migration hop count, $h_m$, as a function of the
network and 1-median subgraph size, $N$ and $\alpha$,
respectively.

When compared with the $10$x$10$ grid, cDSMA's trend to abort
early the migration process only deteriorates with the increase of
network size and diameter--note that rectangular grids feature
larger diameter and, generally, longer (shortest) paths than
equal-size square grids. This is reflected in both the higher
$b_{alg}$ and the slightly increasing yet overly low $h_m$ values
in Table \ref{uni_hops}. Moreover, there is significantly higher
variance in the convergence speed of the algorithm that implies
dependence on the service generation host, \ie the starting point
of the service migration path.
%Note that for given $\alpha$, the
%1-median subgraph grows linearly with $N$.
On the contrary, two remarks can be made as to how the cDSMA
performance scales in B-A graphs: a) its accuracy remains
practically the same as the network size grows; and b) the network
size does not affect the convergence speed of the algorithm, which
needs on average two migration hops to reach a host with
very-close-to-optimal access cost. In other words, even under the
unfavorable hypothesis of uniform service demand, the algorithm
exhibits attractive scalability properties when running over B-A
graphs.

\subsection{Synthetic topologies: experiments under non-uniform demand}\label{subsec:non-uniform}
We repeat our experiments with B-A and grid graphs, only now we
introduce asymmetry in the service demand distribution within the
network.
%namely, some nodes may request a service more
%frequently/intensively than others.
We consider and study separately the two options described in
\ref{sec:evaluation} as to how this asymmetry emerges
\emph{spatially} across the network.

\subsubsection{Spatially random demand distribution}
The service demand is distributed randomly in the network.
Interest in the service may vary but is spread across the network
nodes without any phenomena of spatial concentration. The service
demand asymmetry is modelled by Zipf distributions of variable
skewness parameter values $s$.

\begin{figure}[ht]
%\label{wons}
\begin{minipage}[b]{.55\linewidth}
\centering
\includegraphics[width=6.cm]{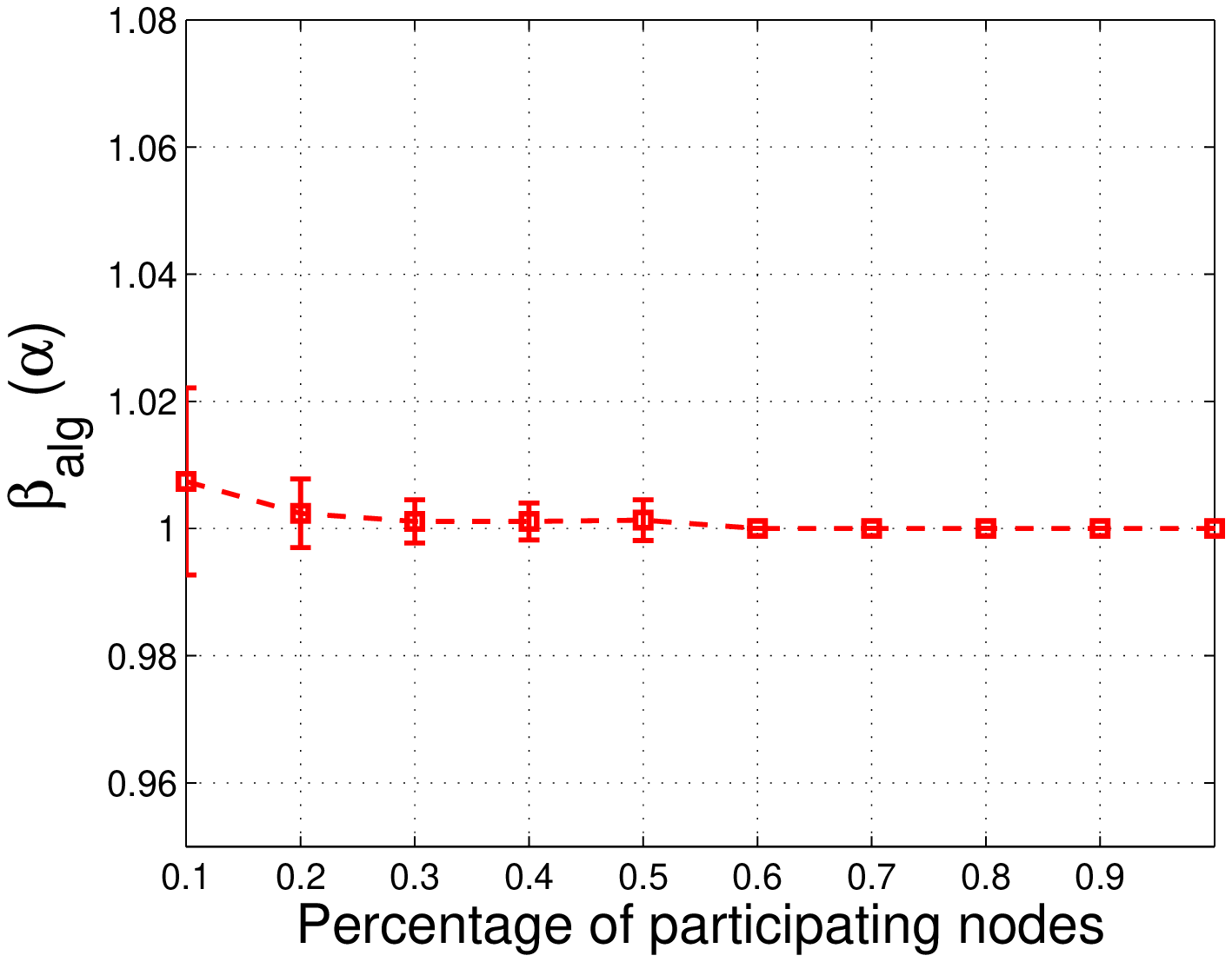}
%\caption[short]{}
\end{minipage}
\hspace{0.35cm}
\begin{minipage}[b]{0.4\linewidth}
\centering
\includegraphics[width=6.cm]{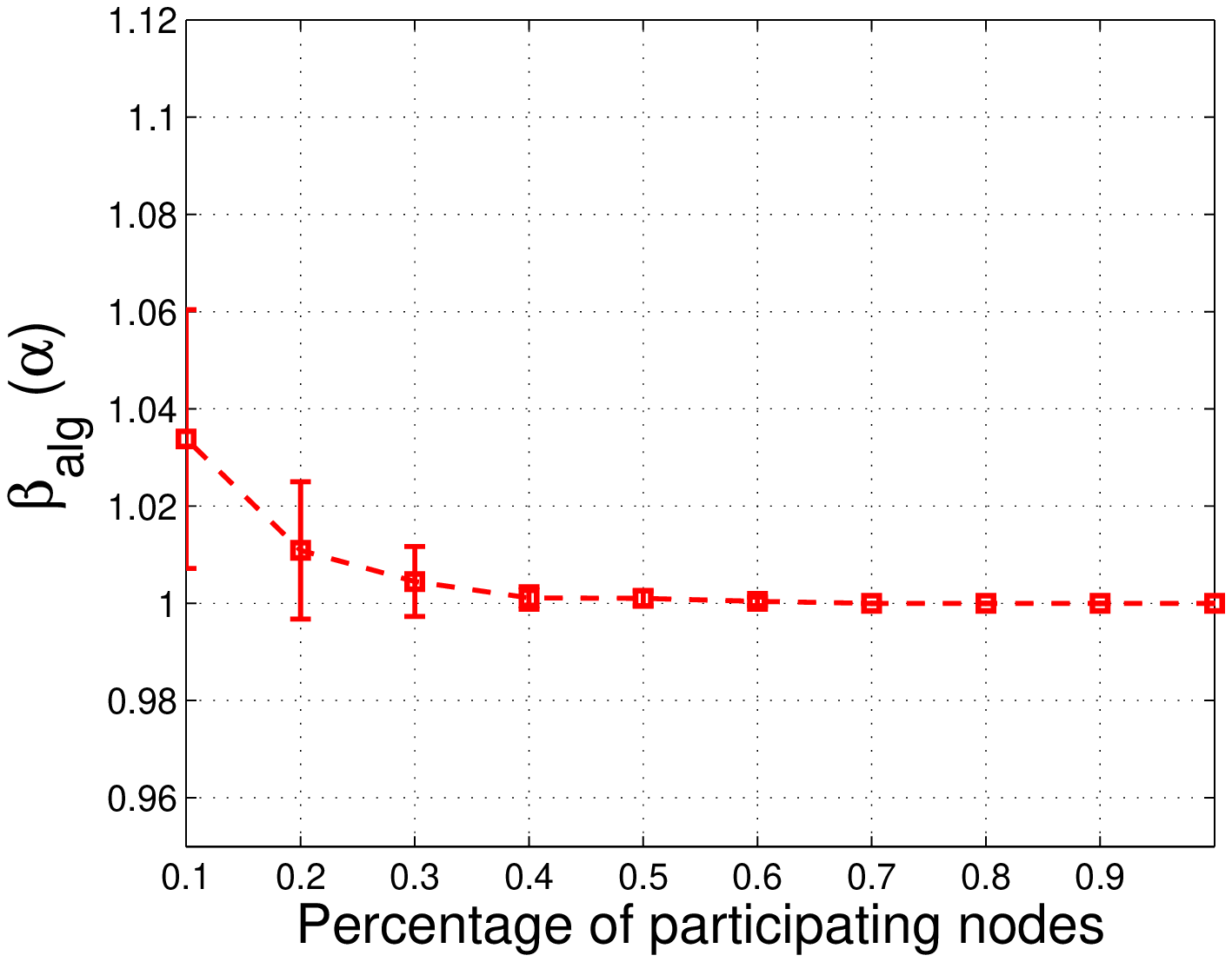}
%\caption[short]{}
\end{minipage}
\caption{cDSMA accuracy for 10x10 nodes B-A graphs(left) and grids(right) as a function
of the 1-median subgraph size: non-uniform service demand
distribution ($s=1$).} \label{fig:synthetic_non_uni}
\end{figure}

\begin{table*}[ht]
\centering \caption{Average normalized excess cost and hopcount for B-A and
Grid networks: non uniform service demand (\lowercase{s}=1)} \label{non_uni_hops} {\scriptsize
\begin{tabular}{ c || c c c c | c c c c}
\multicolumn{1}{c}{}&\multicolumn{4}{c}{\scriptsize{\textbf{B-A graph}}}&\multicolumn{4}{c}{\scriptsize{\textbf{Grid network}}}\\
\cline{1-9}
& & & & & & & & \\
%\hline
\footnotesize{Network size N} &\footnotesize{$\beta_{alg}(0.1)$}&\footnotesize{ $h_m(0.1)$} &\footnotesize{$\beta_{alg}(0.4)$}&\footnotesize{ $h_m(0.4)$} &\footnotesize{$\beta_{alg}(0.1)$}&\footnotesize{ $h_m(0.1)$} &\footnotesize{$\beta_{alg}(0.4)$}&\footnotesize{ $h_m(0.4)$   }\\
\hline
& & & & & & & & \\
50\tiny{(25x2)}&1.0156$\pm$0.0205& 1.60$\pm$0.48&1.0014$\pm$0.0038&1.85$\pm$0.35&1.0083$\pm$0.0068&1.50$\pm$0.37 & 1.0062$\pm$0.0047& 1.10$\pm$0.22 \\
%\hline
%& & & & & & & & \\
100\tiny{(25x4)}& 1.0070$\pm$0.0143&2.15$\pm$0.35&1.0015$\pm$0.0034&1.90$\pm$0.22 &1.0553$\pm$0.0319& 1.35$\pm$0.35& 1.0025$\pm$0.0020&1.15$\pm$0.26  \\
%& & & & & & & & \\
%\hline == ==
200 \tiny{(40x5)}&1.0016$\pm$0.0031&1.90$\pm$0.22&1.0003$\pm$0.0007&2.05$\pm$0.16& 1.0510$\pm$0.0346& 1.47$\pm$0.73& 1.0031$\pm$0.0047& 1.90$\pm$0.65\\
%& & & & & & & & \\
%\hline
300 \tiny{ (60x5)}&1.0029$\pm$0.0068& 2.05$\pm$0.16& 1.0000$\pm$0.0000&2.00$\pm$0.00&  &  & & \\
& & & & & & & & \\
\hline
\end{tabular}}
\end{table*}

Figure \ref{fig:synthetic_non_uni} plots the average normalized
excess cost for $s=1$. Again, the impact on the two types of
synthetic graphs is different. For B-A graphs, the already high
accuracy of cDSMA improves further. It lies within $1\%$ of the
optimal already for $a = 10\%$ and $N=100$ nodes and improves over
the respective values under uniform service demand for all network
sizes. Overall, the demand asymmetry magnifies the existing
attraction forces towards the globally optimal service location,
helping the algorithm to move away from locally optimal, yet
globally suboptimal, hosts. The convergence speed of cDSMA is
practically the same for networks in the range of 100 to 300
nodes.

On the other hand, the algorithm performance over grids is almost
invariable with many entries in Tables \ref{uni_hops} and
\ref{non_uni_hops} remaining practically the same. In fact,
grid-like topologies set a negative benchmark for the performance
of cDSMA requiring far more nodes within the 1-median subgraph
to yield comparable accuracy with B-A graphs for the same network
size and service demand distributions. Or, equivalently, for the
same 1-median subgraph size, it needs a significantly higher
asymmetry in the service demand distribution, as shown more
clearly below.

\subsubsection{Spatially correlated demand distribution}
The service demand now exhibits spatial correlation. Interest in
the service is concentrated in a particular graph neighborhood, as
the case may be when the service has strongly local scope.

We model these scenarios by inserting a cluster of nodes with high
service demand in a random area within a grid. The $K$ cluster
nodes collectively represent some percentage $z\%$ of the total
demand for the service, whereas the other $N-K$ nodes share the
remaining $(100-z)\%$ of the demand. We call the ratio $z/(100-z)$
the demand \emph{spatial contrast} $C_{sp}$. In 2D grids, clusters
are formed by a cluster head node together with its four $1$-hop
($R=1$) or twelve $1$- and $2$-hop ($R=2$) neighbors. The contrast
can then be written as:
\begin{equation}\label{eqn:contrast}
    C_{sp}(R,s) = \frac{\sum_{n=1}^{K} w(n;s,N)}{\sum_{n=K+1}^N w(n;s,N)} = \frac{\sum_{n=1}^K 1/n^s}{\sum_{n=K+1}^N 1/n^s}
\end{equation}

and the average normalized excess cost becomes a function of both
$\alpha$ and the contrast value.
%$\beta_{alg}(\alpha, C_{sp})$.

%Without loss of generality we focus on the case in which the
%cluster of high demands is -somehow- limited within the radius of
%$1$ hop; in the Table~\ref{tbl:contrast} we present the normalized
%excess cost with and without spatial correlation of demands, in a
%10x10 grid topology. Having the top demand values stemming from a
%certain neighborhood we actually ``produce'' a single pole of
%strong attraction for the migrating service which is -now- capable
%of following the demand gradient more effectively than before.
%When the percentage of the total demand held by the cluster nodes
%grows larger, resulting higher spatial contrast, the pole gets
%even more intense leading the service firmly to the optimal
%location.

The values of $\beta_{alg}(\alpha, C_{sp})$ under spatially random
and correlated ($R=1$) distribution of demands are reported in
Table~\ref{tbl:contrast} for a $10$x$10$ grid topology. Having the
top demand values stemming from a certain network neighborhood we
actually ``produce'' a single pole of strong attraction for the
migrating service. Our algorithm now follows the demand gradient
more effectively than before. As the percentage of the total
demand held by the cluster nodes grows larger, resulting in higher
spatial contrast, the pole gets even stronger driving the service
firmly to the optimal location.

\begin{table}[tbp]
\centering \caption{Average normalized excess cost under spatially
correlated service demand} \label{tbl:contrast} {\footnotesize
\begin{tabular}{c ||  c  c  c}
\hline 
& & & \\
skewness $s$ & $C_{sp}(1,s)$ &  $\beta_{alg}(0.1)$ &
$\beta_{alg}(0.1,
C_{sp})$ \\
& & & \\
\hline
& & & \\
 1  &  0.786   &    1.035$\pm$0.027& 1.016$\pm$0.023\\
 2  &  8.540   &    1.003$\pm$0.006& 1.0$\pm$0.0\\
& & & \\
\hline
\end{tabular}}
\end{table}

It follows that $R=2$ and higher service demand distribution
asymmetry $s$ only sharpen the spatial demand contrast,
concentrating more the demand in space; namely, 61\% of the
service demand is spread across $13$ nodes for ($s=1, R=2$) and
89\% across five nodes for ($s=2, R=1$). The attractive forces
applied on the migrating service grow so that the algorithm finds
easier its way towards the optimal location.

\subsection{Experiments on real-world network topologies}
\label{subsec:data_runs}

%Moving forward to the main evaluation part of our work, we seek to
%obtain further insights on the performance, as well as explore the
%applicability, of our heuristic on real-world networks;
%furthermore, the results coming from this section may provide good
%support to the introduction's case study. To this end we have
%performed a wide range of simulations using a Tier-1 ISP
%topologies dataset\cite{dataset}.

%Having shown and discussed the behavioral dynamics of the
%algorithm over synthetic topologies, we turn our attention to
%real-world networks.
Real-world networks do not typically have the predictable
structure and properties of B-A graphs and grids and may differ
substantially the one from another. Nevertheless, we show below
that insightful analogies can be drawn between these networks and
the B-A and grid topologies regarding the behavior of our service
placement mechanism.

The ISP topology dataset includes 264 Tier-1, 244 Transit, and 342
stub ISP network topology files. They represent snapshots of 14
different ISP network topologies, as measured at different time
epochs within the interval 2004-2008. We have focused on the
larger Transit- and Tier-1 ISP datafiles, with topology sizes
ranging from $100$ to $1000$ nodes, approximately. We chose to
identify and primarily work with datasets, where the size of the
maximal connected component, to be denoted by $mCC$, approaches
the full vertex set of the measured graph~\footnote{Many of the
original network topology files that have been released miss some
edges, resulting in more than one connected components. The
measurement inaccuracies are mainly due to filtering incurring in
the ISP borders or ISP hardware updates.}. The connected
components for each topology are retrieved via the well-known
linear-time algorithm of Karp and Tarjan ~\cite{tarzan}. Herein we
present and discuss results from a representative subset of the
datasets we experimented with, as shown in \ref{tbl:data_info}.
They correspond to snapshots of four Tier-1 and three Transit ISP
networks and were chosen so that there is adequate
%Table~\ref{tbl:data_info} lists a representative subset
%of the topologies we experimented with, chosen deliberately to
variance in size, diameter, and connectivity degree statistics.

%************************************** to be included in the paper
% A more exhaustive evaluation of our algorithm over a broader set of
% these topology snapshots is reported in \textbf{[ref:techRep]}.
%**************************************************************************

\begin{table}[ht]
\centering \caption{Selected ASes} \label{tbl:data_info}
{\footnotesize
\begin{tabular}{c c||  c  c  c}
\hline
Type & Dataset id &AS Number   & Name &  Extracted on\\
\hline
       &     &          &                     &              \\
Tier-1 & 36  &  3549    &   Global Crossing   &   2006-05-03 \\
       & 35  &   -//-   &       -//-          &   2006-07-13\\
       & 33  &  2914    &     NTTC-Gin        &   2008-12-03 \\
       & 23  &  1239    &     Sprint          &   2008-09-30 \\
       & 21  &  1239    &     -//-            &   2008-08-27 \\
       & 27  &  3356    &     Level-3         &   2004-09-24 \\
       & 13  &   -//-   &       -//-          &   2005-03-17 \\
\hline
       &  &     &                       &              \\
Transit& 46  &  3292    &      TDC            &   2008-05-01\\
       & 41  &   680    &     DFN-IPX-Win     &   2006-05-03\\
       & 40  &   786    &     JanetUK         &   2008-07-01 \\
       &  &     &                       &              \\
\hline
\end{tabular}}
\end{table}

%\subsubsection{Results}
Table~\ref{ISP_results} summarizes the performance of cDSMA over
the real-world topologies. The listed results include the minimum
number of nodes $|G^{i}|$ required to achieve a solution that lies
within $2.5\%$ of the optimal and the average migration hop count
$h_m$ for different levels of asymmetry in the service demand
distribution.

The main observation is that both $\alpha_{0.025}$ and $|G^{i}|$
show a remarkable insensitivity to both topological structure and
service demand dynamics. Although the considered ISP topologies
differ significantly in size and diameter, the number of nodes we
need to include in the $1$-median problem solution does not
change. On the contrary, around half a dozen nodes suffices to get
good accuracy even under uniform demand distribution, the least
favorable scenario for our algorithm as discussed in Sections
\ref{subsec:uniform} and \ref{subsec:non-uniform}. Likewise,
$\alpha_{0.025}$ and $|G^{i}|$ remain practically invariable with
the demand distribution skewness. Although for larger values of
$s$, few nodes exhibit asymmetrically large demand values and
become stronger attractors for the algorithm, the added value for
the algorithm accuracy is negligible.

This two-way insensitivity of our algorithm bears two significant
implications for its more practical implementation aspects.
Firstly, the computational complexity when solving instances of
the 1-median problem can be negligible and scales well with the
size and diameter of the network. Secondly, the algorithm
performance is robust to possibly inaccurate estimates of the
service demand each node poses.

%To be first observed is the number of nodes needed to effectively
%(i.e having $\epsilon$ low) solve the placement problem on
%topologies taken from different ISPs and covering a wide range of
%values with respect to their structural characteristics (i.e.
%diameter, number of links or size). Nevertheless, the nodes
%compromising the optimization's solution space are steadily kept
%below a dozen even for networks consisting of more than 300 nodes
%(i.e. an average Tier-1 size) and exhibiting variance in diameter
%within an order of magnitude. Moreover, having the parameter $s$
%grown larger than zero, results in a few nodes exhibiting very
%large demand values but actually brings about negligible influence
%on the $\alpha_{0.025}$ value. This important characteristic of
%our algorithm allows for significant computational benefits to be
%drawn by achieving an efficient solution with a limited input,
%regardless demand shifts. Otherwise, as already stressed, dealing
%with large optimizations (even under demand uniformity) comes with
%a great cost if not being unfeasible.

A last remark is appropriate with respect to the topological
structure of these real-world topologies. The equally well
algorithm behavior under uniform demand distribution ($s=0$)
suggests that there is already adequate structure in the network
topology. As the probability distribution of the connectivity
degree in these networks suggests (see
Fig.~\ref{fig:ISP_degree_distr}), there are high-degree nodes and
considerable variance in the connectedness properties of nodes
across the network. In fact, the high-degree nodes serve in a way
similar to the high-degree nodes in B-A graphs; they are easily
``identifiable'' by our algorithm as low-cost hosts for the
migrating service and, even for small 1-median subgraph sizes,
their attraction forces are strong enough to pave a cost-effective
service migration path.

%========== ISP topologies degree distribution
\begin{figure*}[ht]
\centering
 \subfigure[\scriptsize{dataset 36}]{
\includegraphics[scale=0.26]{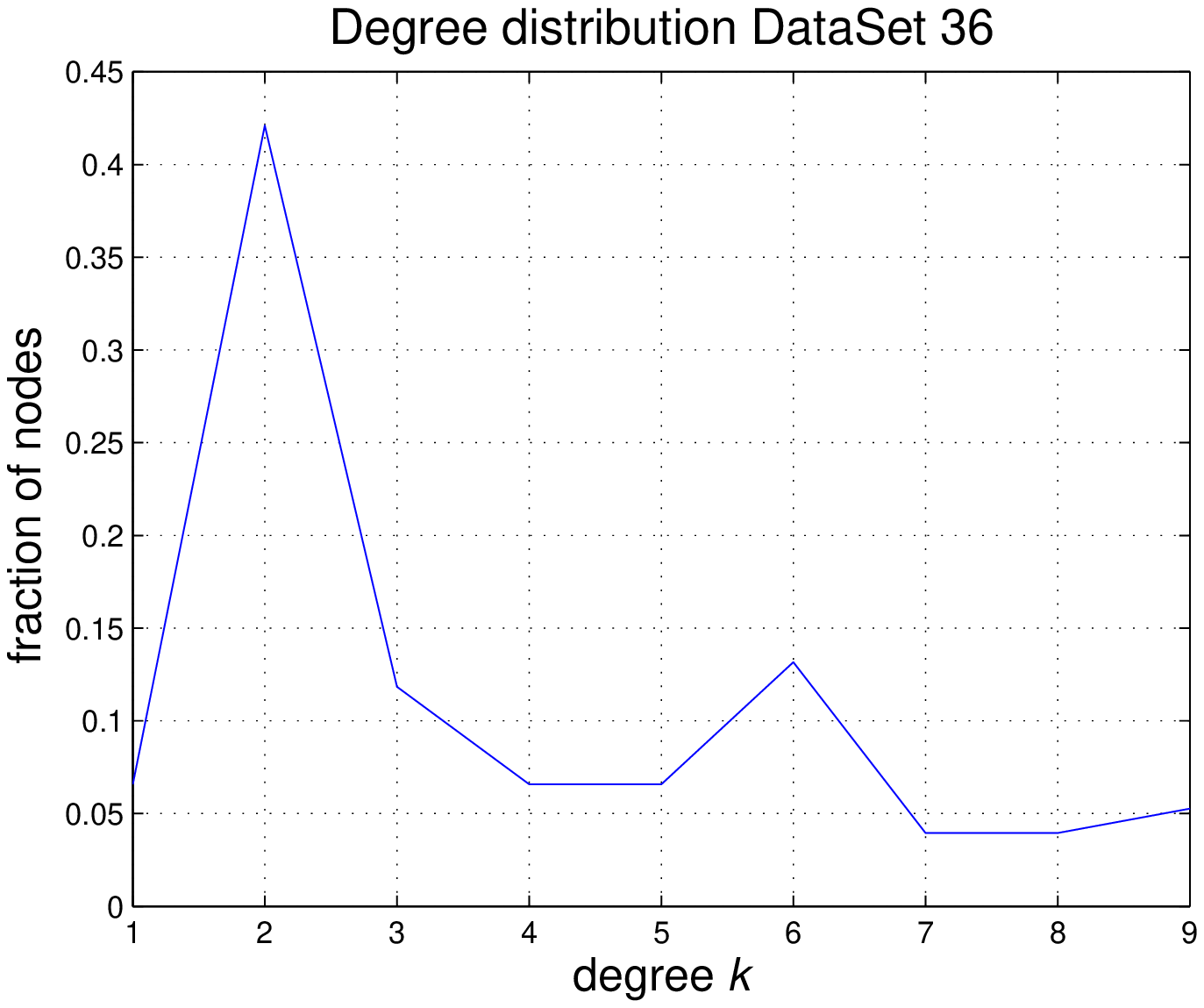}
\label{fig:DataSet_36} }
 \subfigure[\scriptsize{dataset 35}]{
\includegraphics[scale=0.26]{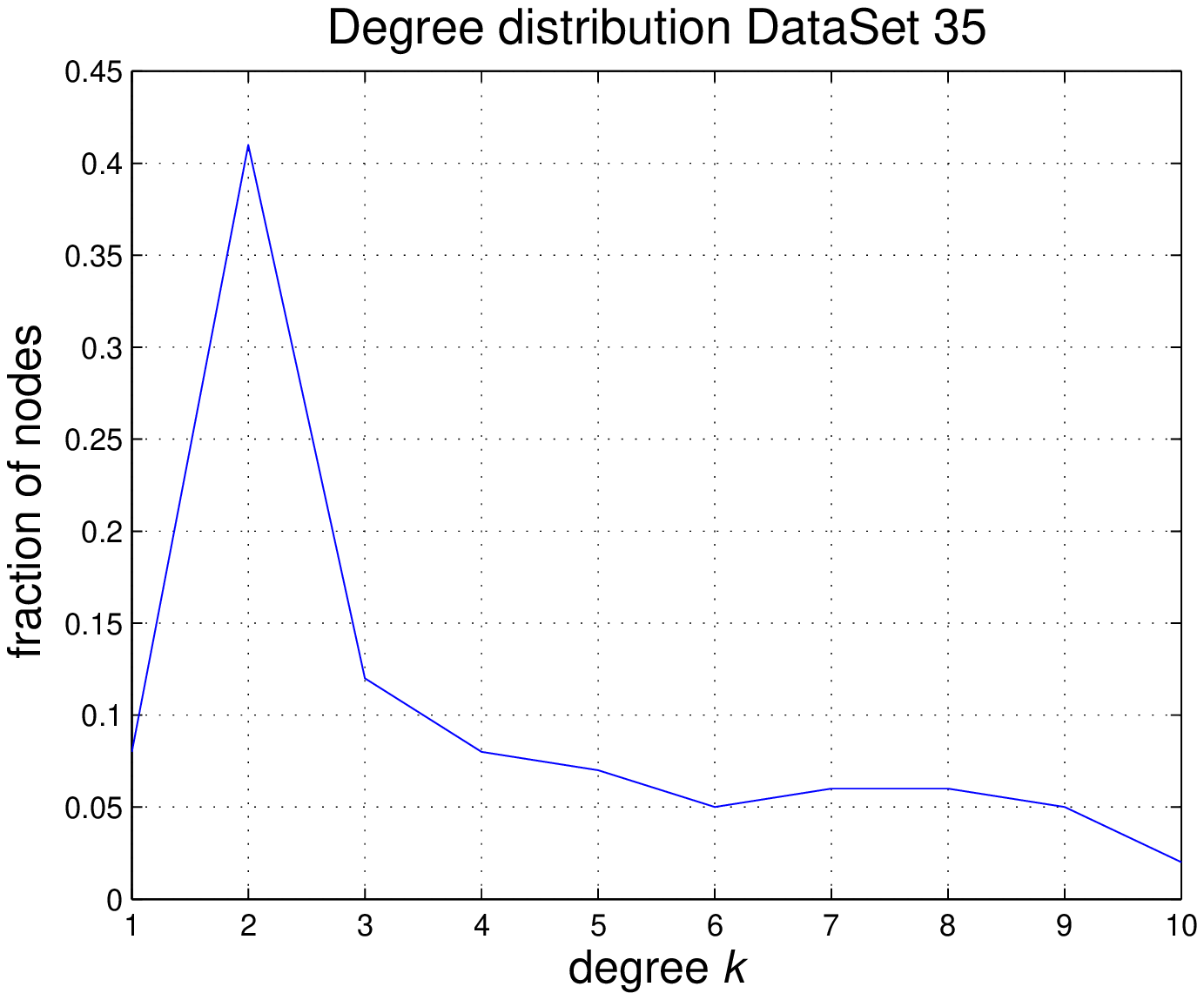}
\label{fig:DataSet_35} }
 \subfigure[\scriptsize{dataset 33}]{
\includegraphics[scale=0.26]{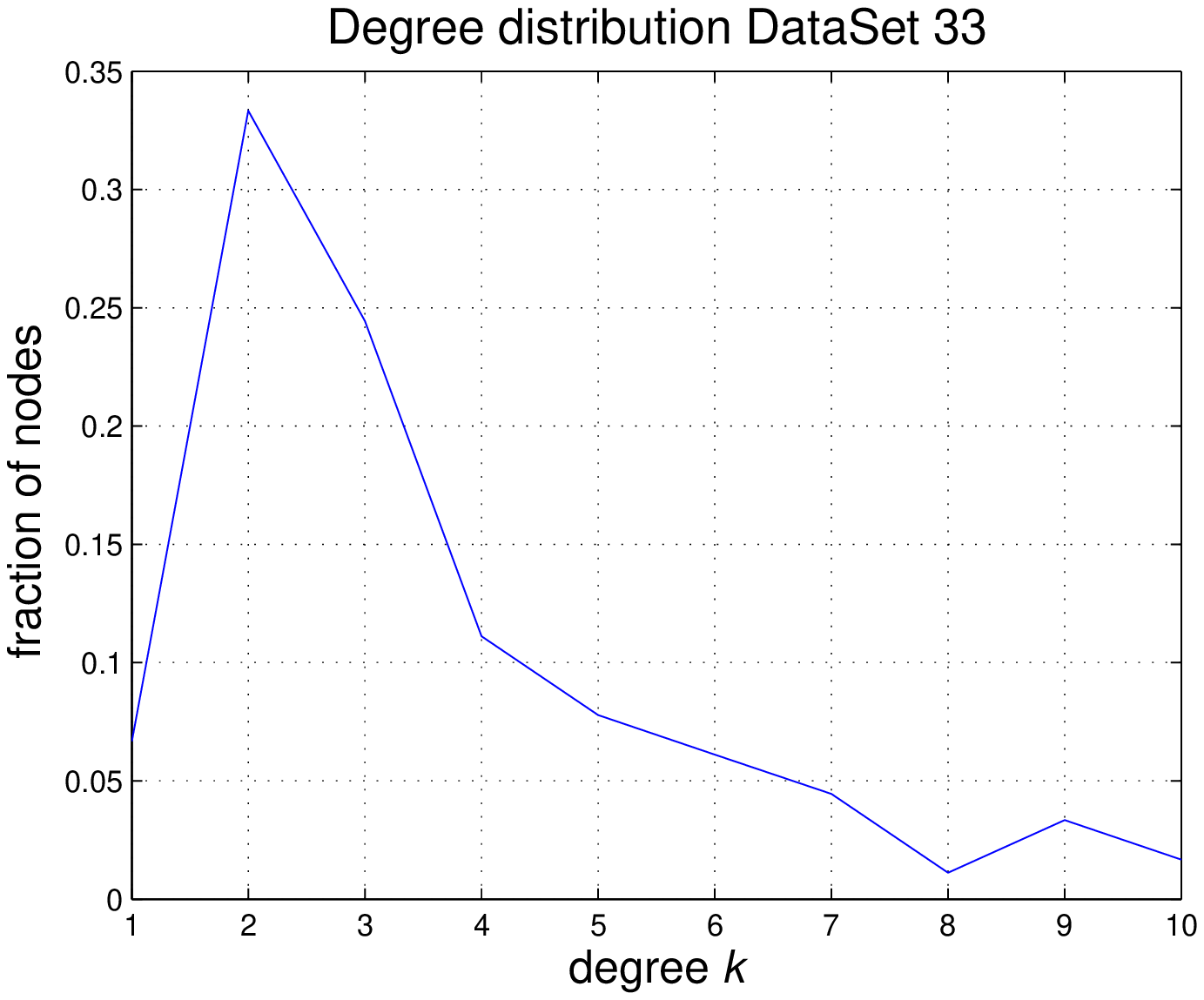}
\label{fig:DataSet_33} }
 \subfigure[\scriptsize{dataset 23}]{
\includegraphics[scale=0.26]{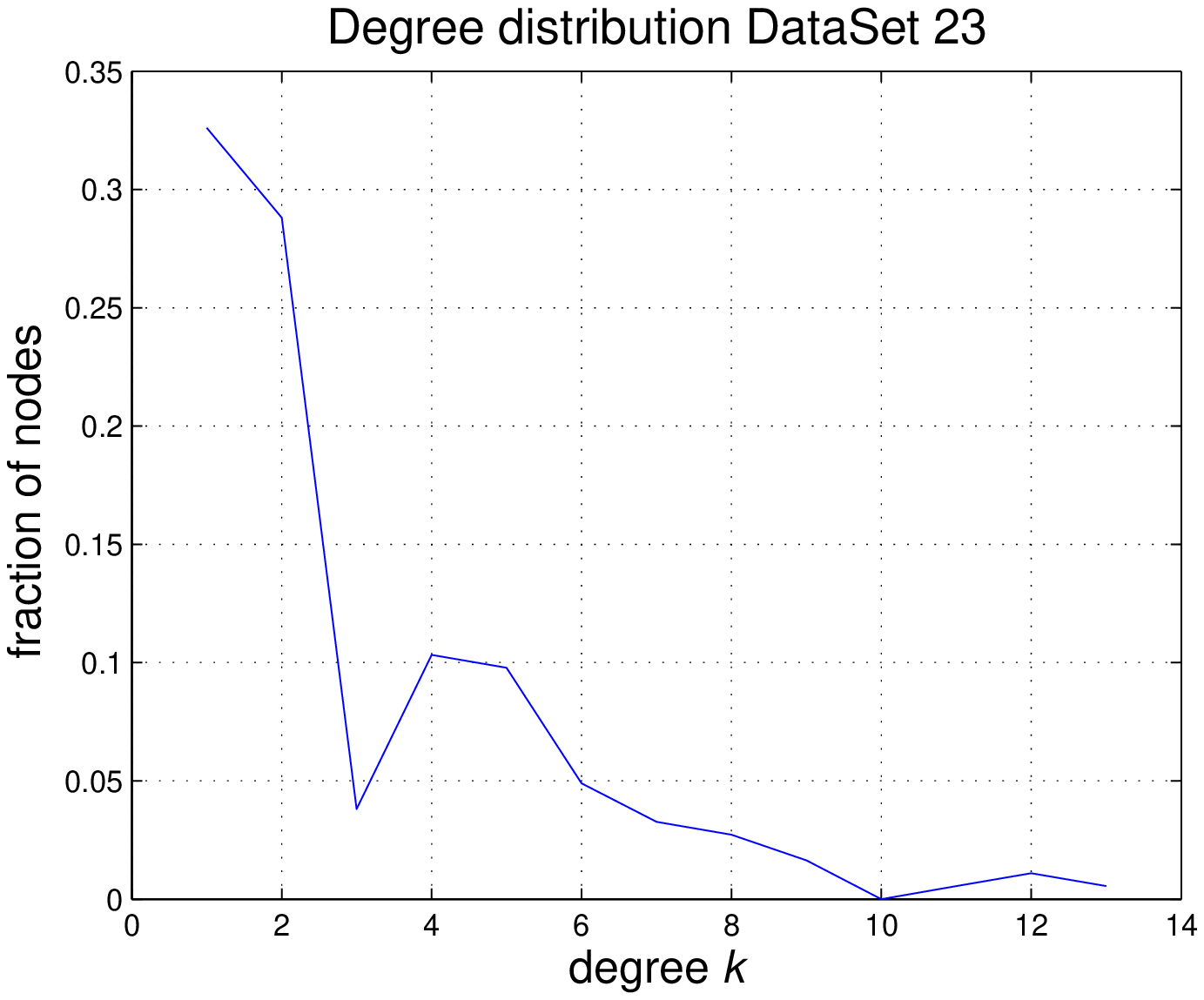}
\label{fig:DataSet_23} }
 \subfigure[\scriptsize{dataset 21}]{
\includegraphics[scale=0.26]{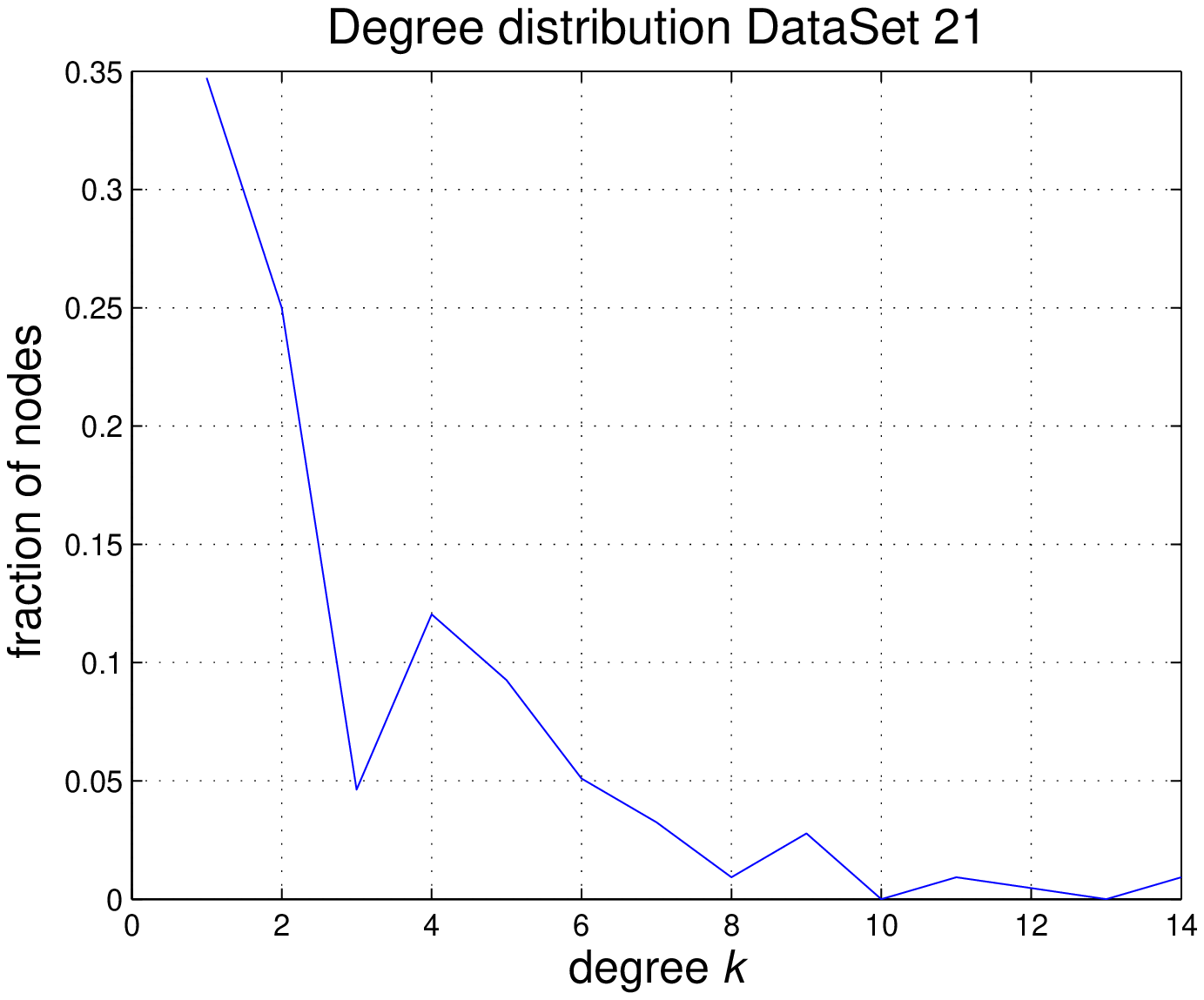}
\label{fig:DataSet_21} }
 \subfigure[\scriptsize{dataset 27}]{
\includegraphics[scale=0.26]{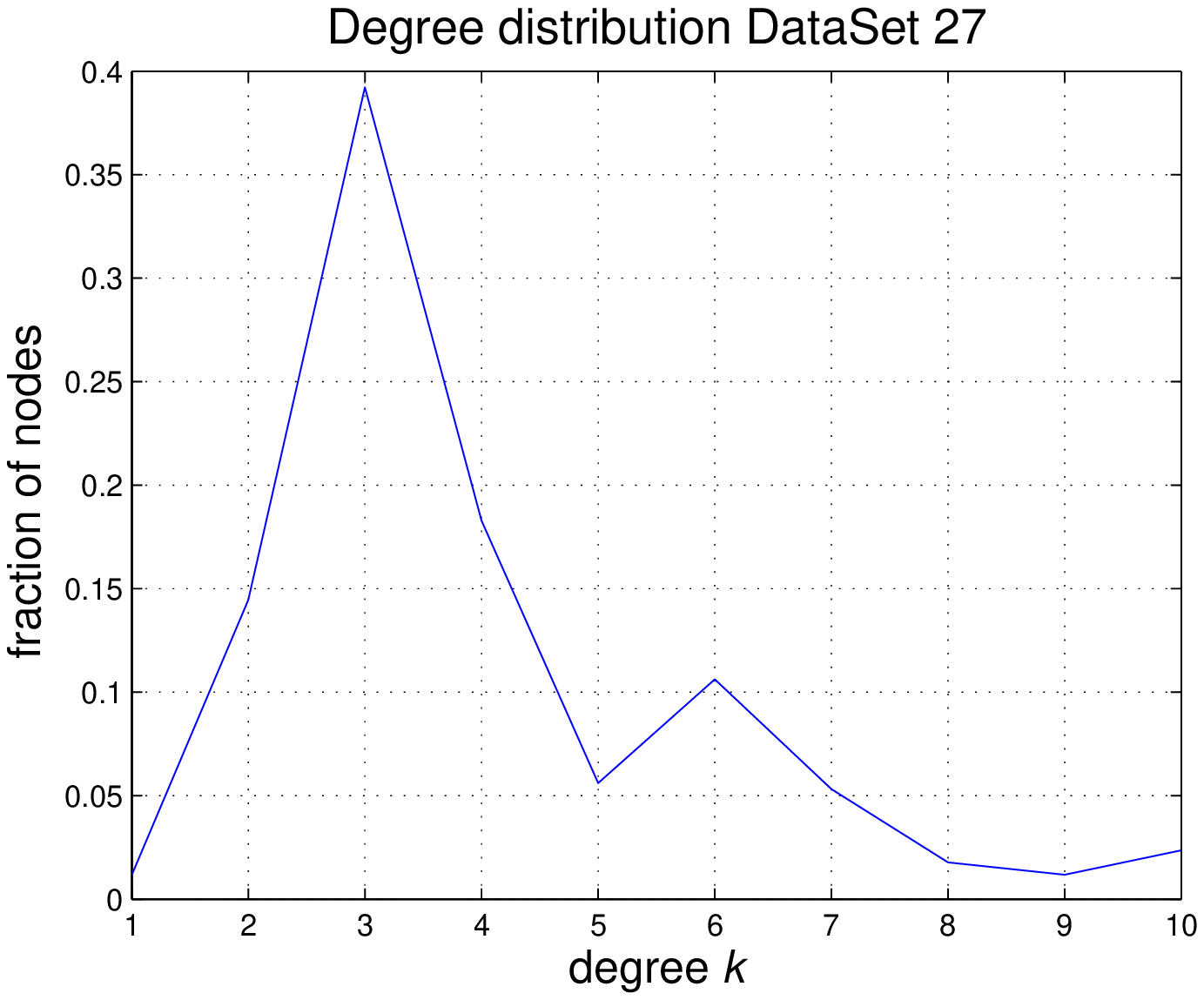}
\label{fig:DataSet_27} }
 \subfigure[\scriptsize{dataset 13}]{
\includegraphics[scale=0.26]{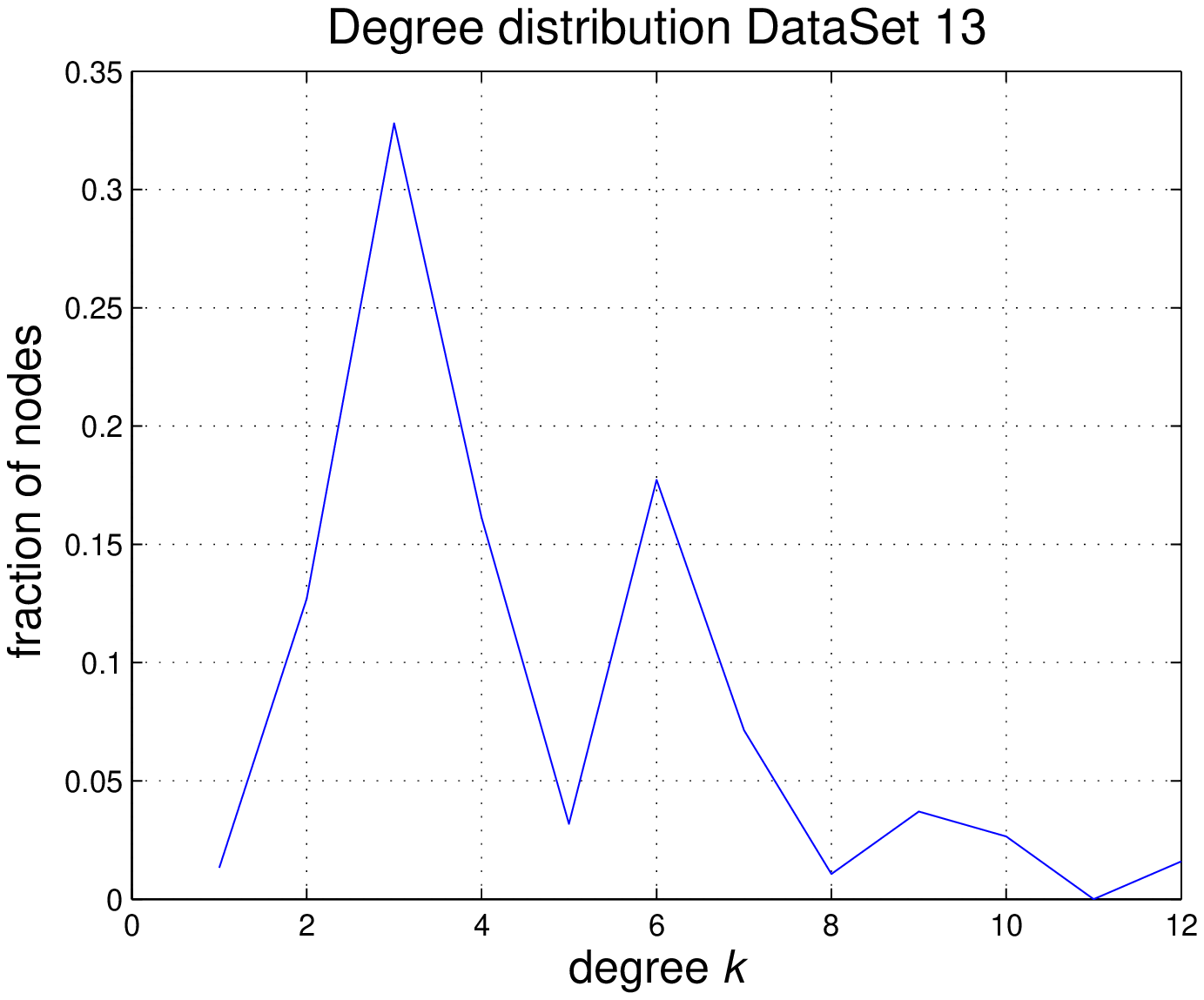}
\label{fig:DataSet_13} }
 \subfigure[\scriptsize{dataset 46}]{
\includegraphics[scale=0.26]{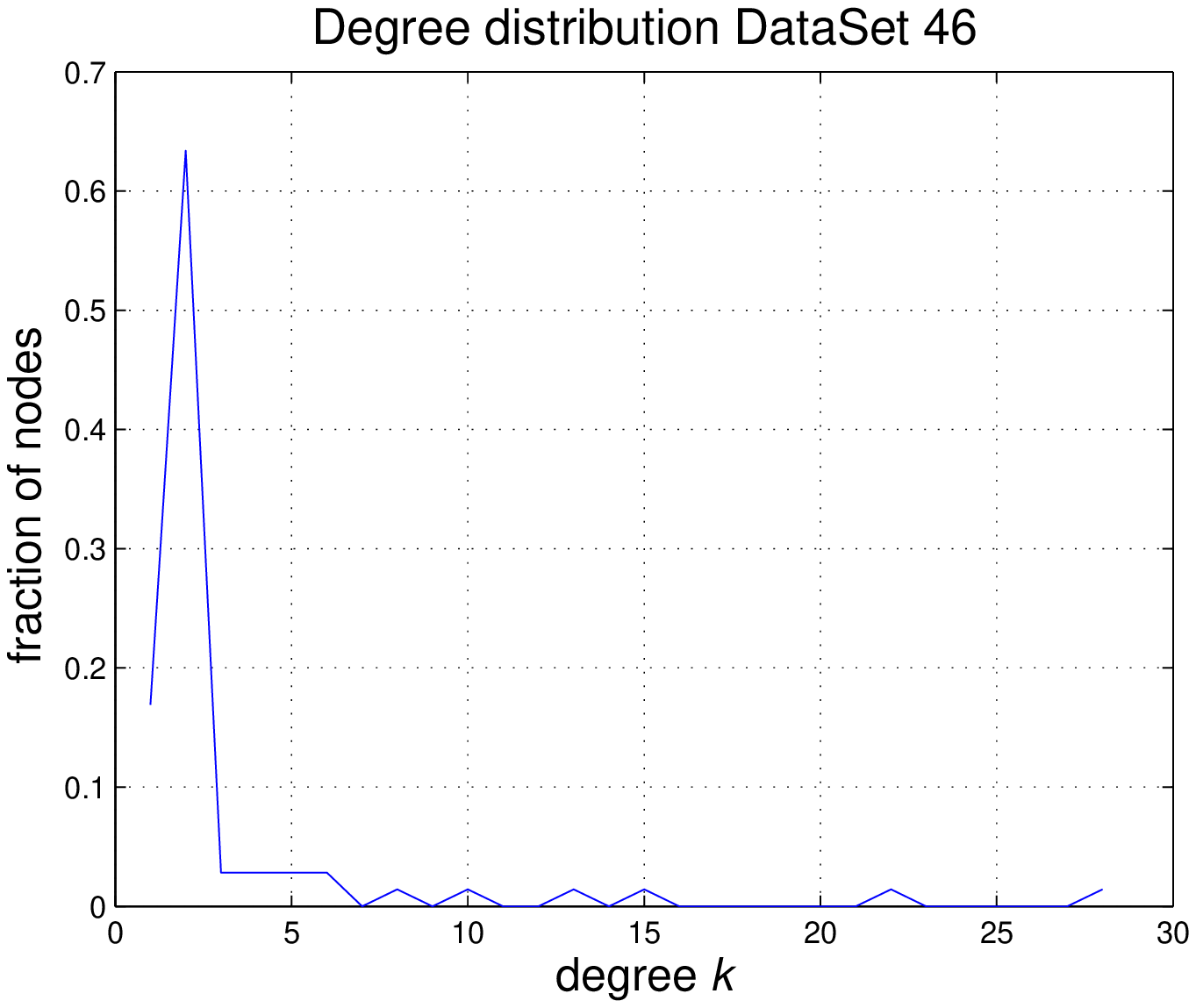}
\label{fig:DataSet_46} }
 \subfigure[\scriptsize{dataset 41}]{
\includegraphics[scale=0.26]{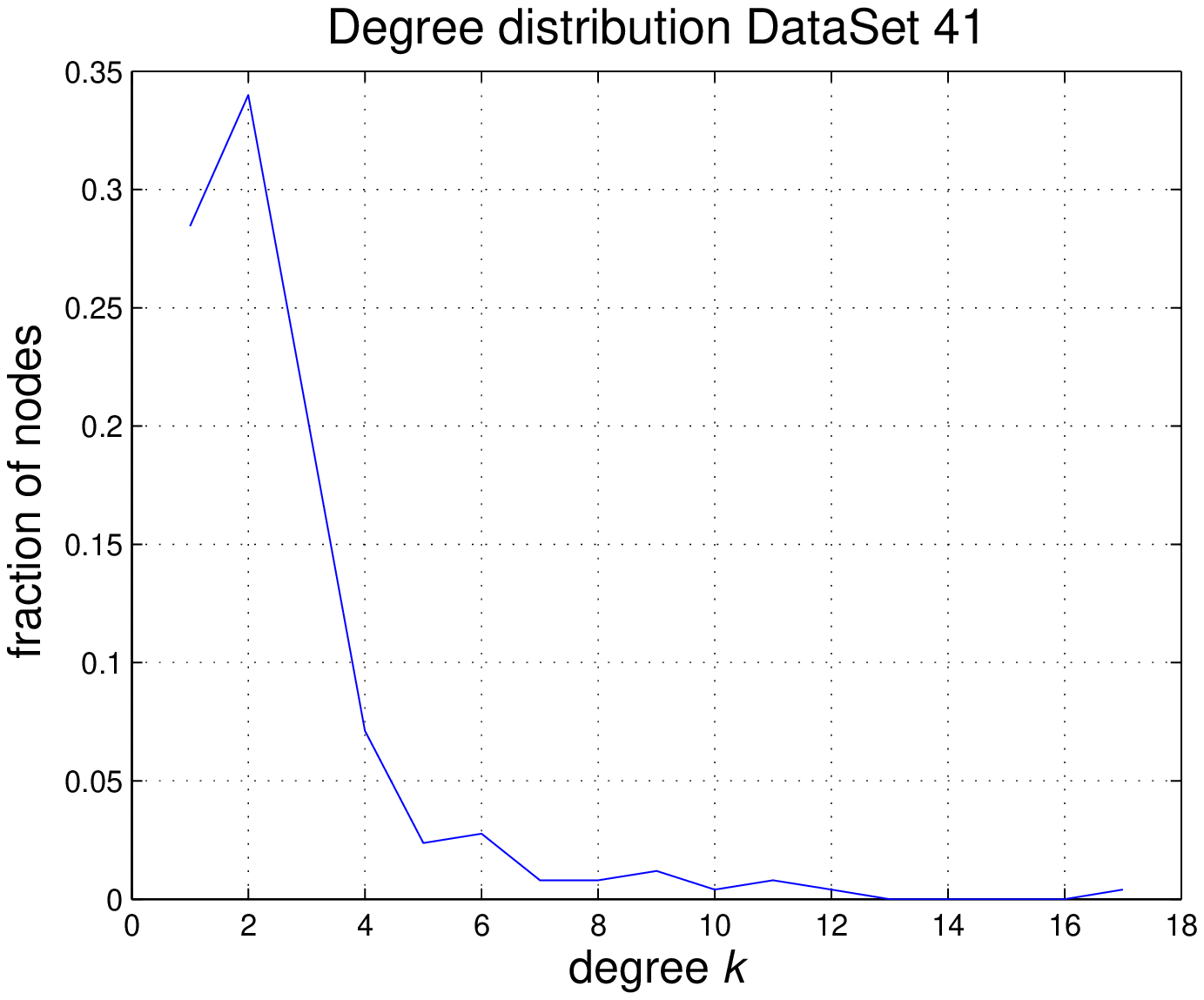}
\label{fig:DataSet_41} }
 \subfigure[\scriptsize{dataset 40}]{
\includegraphics[scale=0.26]{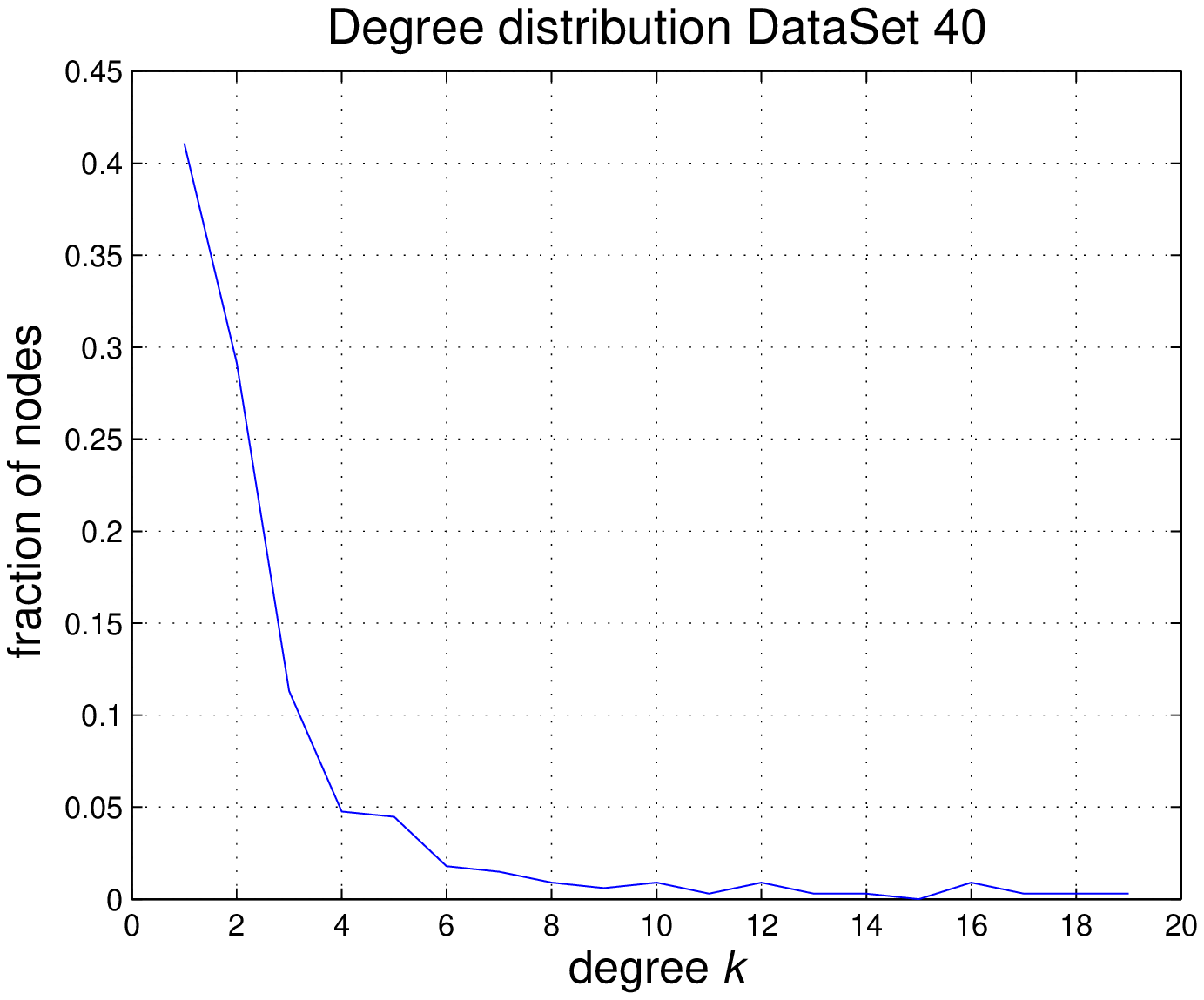}
\label{fig:DataSet_40} }
\caption[Optional caption for list of figures]{Degree distributions of the ISP topology snapshots}
\label{fig:ISP_degree_distr}
\end{figure*}

\begin{table*}[ht]
\centering \caption{Mean value of $\alpha_{\epsilon}$ for various
datasets under different demand distributions} \label{ISP_results}
{\scriptsize
\begin{tabular}{c c c c c || c c |c c |c c}
\multicolumn{5}{c}{}&\multicolumn{2}{c}{\scriptsize{\textbf{s=0}}}&
\multicolumn{2}{c}{\scriptsize{\textbf{s=1}}}&
\multicolumn{2}{c}{\scriptsize{\textbf{s=2}}}\\
\cline{1-11}
& & & & & & & & & & \\
%\hline
Type&Dataset id &mCC nodes& Diameter&<Degree>&\scriptsize{$
\alpha_{0.025}$}& $\lceil|G^{i}|\rceil$& \scriptsize{$
\alpha_{0.025}$}&$\lceil |G^{i}|
\rceil$& \scriptsize{$\alpha_{0.025}$}&$\lceil
|G^{i}| \rceil$\\
\hline
& & & & & & & & & & \\
\textbf{Tier-1}&36 & 76 & 10 & 3.71&0.047$\pm$0.001&4&0.047$\pm$0.002
&4&0.046$\pm$0.001&4\\
%\hline
%& & & & & & & & & & & \\
&35&100&9 & 3.78&0.045$\pm$0.002&5&0.045$\pm$0.001&5&0.043$\pm$0.001&5\\
%& & & & & & & & & & & \\
%\hline == ==
&33&180&11&3.53&0.024$\pm$0.002&5&0.022$\pm$0.002
&4&0.019$\pm$0.002&4\\
%& & & & & & & & & & & \\
%\hline
&23&184&13&3.06&0.019$\pm$0.002&4&0.018$\pm$0.002&4&0.017$\pm$0.002&4\\
%& & & & & & & & & & & \\
%\hline
&21&216&12&3.07&0.016$\pm$0.002&4&0.016$\pm$0.002&4&0.014$\pm$0.003&4\\
%& & & & & & & & & & & \\
&27&339&24&3.98&0.018$\pm$0.002&7&0.017$\pm$0.002&6& 0.014$\pm$0.003&5\\
&13&378&25&4.49&0.012$\pm$0.002&5 &0.012$\pm$0.002&5& 0.011$\pm$0.002&5\\
\hline
                &  &  &   &    &               &  &               & &                &       \\
\textbf{Transit}&46&71&  9&3.30&0.033$\pm$0.003&3 &0.027$\pm$0.004&2& 0.026$\pm$0.003&2\\ %
                &41&253&14&2.62&0.019$\pm$0.003&5 &0.015$\pm$0.003&4&0.015$\pm$0.003&4\\
                &40&336&14&2.69&0.012$\pm$0.003&5 &0.012$\pm$0.002&5 &0.013$\pm$0.002&5\\
%& & & & & & & & &
\hline
\end{tabular}}
\end{table*}
%\textbf{PENDING:}
%\begin{enumerate}
%\item explain in brief the way the data have been collected
%\item \textbf{put forward arguments to defend the appropriatenece of the data for testing our heuristic}
%\item  illustrate the effectiveness of our approach in these topologies - plot $\alpha^{+}=F(s)$ for different data sets?
%\item built a table showing the mean hopcount to optimal for realistic topologies under power-law demanding - could we have the average over a set of topologies?
%\end {enumerate}

\subsection{cDSMA vs. locality-oriented service migration}
\label{subsec:positioning} The way cDSMA determines the service
migration path
%Although
%intuitively sound, having
%by selecting the candidate future service host nodes
%selected
%according to their ability to concentrate demand flow towards the
%current service host,
clearly differentiates from typical ``local-search'' approaches.
%at least not from a spatial standpoint.
Local-search solutions such as the R-ball heuristic
in~\cite{SLOSB-DSMSISD-10}, for example, restrict \emph{a-priori}
their search for a better service host to the neighborhood of the
current service location. On the contrary, cDSMA focuses its
search for the next service host in certain directions. Nodes
lying across a (shortest) path, which serves many requests for a
service, exhibit relatively high $wCBC$ values. The resulting
1-median subgraph is spatially stretched across that path and
therefore oversteps the local neighborhood ``barriers''.

% This property assisted by topological aspects provides the advantage of long ``jumps'' for the service, or equivalently small migration hop-count $h_{m}$
% While there is a clear trade-off between the neighborhood size (to search, in each step, for a better solution than the current one) and the precision of the outcome, in our case the important thing turns out to be the shape of that neighborhood. It is likely that the $wCBC$ values of the nodes lying across a high demand (shortest) path providing the advantage of long ``jumps'' for the service, or equivalently small migration hop-count $h_{m}$.

To compare the above two approaches, we have implemented a
Locality-Oriented Migration heuristic, hereafter abbreviated to
LOM. In LOM we solve the 1-median problem within the direct
neighborhood of $R$ hops around the current host and apply the
same demand mapping mechanism (\ref{subsec:mapping}) to capture
the demand load from nodes lying further than $R$ hops away from
the current service host.
%Being especially interested in the
%number of hops the service executes, we evaluate the two
%approaches knowing a-priori the distance that should be traversed
%to the optimal location.
The comparison of the two approaches for each ISP topology
snapshot in Table~\ref{tbl:comparison} proceeds as follows. We
first generate asymmetric service demand (Zipf distribution with
$s=1$) across the network. We compute the globally optimal service
host node and we select a fixed set of service generation nodes,
at $D_{gen}$ hops away from the optimal service location. We then
calculate the values of $h_{m}$ and $\beta_{alg}$
metrics~\footnote{The void entries are due to the fact that the
most distant node to the global minimum location, lies at some
distance smaller than the $D_{gen}$ value; a piece of information
not captured by the diameter value.} for the two approaches,
cDSMA and LOM. For cDSMA, we have set the parameter
$\alpha=3\%$, meaning that the 1-median subgraph size ranges from
6 to 12 nodes for the networks listed in
Table~\ref{tbl:comparison}.

\begin{table*}[ht]
\centering \caption{convergence speed and accuracy comparison between LOM and cDSMA on realistic topologies} \label{tbl:comparison} {\scriptsize
\begin{tabular}{ c || c |c |c| c|| c| c |c |c|| c| c |c |c|| c| c |c |c}
\hline
\multicolumn{1}{c||}{}&\multicolumn{4}{c||}{\scriptsize{\textbf{Dataset 23}}}&\multicolumn{4}{c||}{\scriptsize{\textbf{Dataset 33}}}&\multicolumn{4}{c||}{\scriptsize{\textbf{Dataset 27}}}&\multicolumn{4}{c}{\scriptsize{\textbf{Dataset 13}}} \\
%\cline{1-17}
%& & & & & & & & & & & & & & & & \\
\multicolumn{1}{c||}{\footnotesize{$D_{gen}$}}  &\multicolumn{2}{c|}{\scriptsize{LOM}}&\multicolumn{2}{c||}{\scriptsize{cDSMA}}&\multicolumn{2}{c|}{\scriptsize{LOM}}&\multicolumn{2}{c||}{\scriptsize{cDSMA}}&\multicolumn{2}{c|}{\scriptsize{LOM}}&\multicolumn{2}{c||}{\scriptsize{cDSMA}}&\multicolumn{2}{c|}{\footnotesize{LOM}}&\multicolumn{2}{c}{\scriptsize{cDSMA}}\\
%& & & & & & & & & & & & & & & & \\
\multicolumn{1}{c||}{}       &\multicolumn{1}{c|}{\textbf{\scriptsize{$h_{m}$}}}&\multicolumn{1}{c|}{\textbf{\scriptsize{$\beta_{alg}$}}}&\multicolumn{1}{c|}{\tiny{$h_{m}$}}&\multicolumn{1}{c||}{\tiny{$\beta_{alg}(3\%)$}}&\multicolumn{1}{c|}{\tiny{$h_{m}$}}&\multicolumn{1}{c|}{\tiny{$\beta_{alg}$}}&\multicolumn{1}{c|}{\tiny{$h_{m}$}}&\multicolumn{1}{c||}{\tiny{$\beta_{alg}(3\%)$}}&\multicolumn{1}{c|}{\tiny{$h_{m}$}}&\multicolumn{1}{c|}{\scriptsize{$\beta_{alg}$}}&\multicolumn{1}{c|}{\tiny{$h_{m}$}}&\multicolumn{1}{c||}{\tiny{$\beta_{alg}(3\%)$}}&\multicolumn{1}{c|}{\tiny{$h_{m}$}}&\multicolumn{1}{c|}{\tiny{$\beta_{alg}$}}&\multicolumn{1}{c|}{\tiny{$h_{m}$}}&\multicolumn{1}{c}{\tiny{$\beta_{alg}(3\%)$}}\\
\hline
 3&  1& 1.1050 & 2 &1 & 1 & 1.0308& 2 &1      & 1 & 1.1109  & 1 &1.0057 &1 &1.1054 &1&1 \\
 4&  1& 1.1275 & 3 &1 & 1 & 1.3206& 2 &1      & 1 & 1.2523  & 1 &1.0057 &1 &1.2312 &1&1  \\
 5&  1& 1.1632 & 2 &1 & 1 & 1.2800& 1 &1.2800 & 2 & 1.1109  & 1 &1      &1 &1.0434 &2&1 \\
 7&  1& 1.6060 & 2 &1 & 3 & 1.0308& 1 &1.0308 & 3 & 1.1763  & 1 &1      &1 &1.4202 &1&1 \\
10& --& --     &--&--& --&   --  &-- &--      & 1 & 1.7094  & 2 &1      &1 &1.4604 &2&1      \\
13& --&--      &--&--& --&   --  &-- &--      & 2 & 1.8579  & 1 &1.0057 &3 &1.6887 &1&1.1054 \\
\hline
\end{tabular}}
\end{table*}

Our expectation before the experiments was that the LOM heuristic
would be characterized by overly higher number of migration hops
since the latter is lower bounded by $D_{gen}/R$ when the service
reaches the globally optimal location. Nevertheless, and
interestingly enough, the LOM approach combines high excess costs
with generally small number of migration hops, irrespective of the
service generation location and for all topologies. Selecting
``blindly'' the $R$-hop neighbors of the current service host as
future candidate hosts, LOM effectively introduces noise to the
mechanism's effort to detect the cost-effective service migration
direction. With LOM the nodes in the 1-median subgraph are spread
more unidirectionally around the service host and the demand
mapping process projects more uniformly the demand contributions
of the remaining nodes on them. Consequently, the migrating
service gets easily trapped in some local minimum, which forces
the migration process to stop too early to achieve an efficient
solution. This resembles the behavior of the cDSMA in grids
under uniform demand. There it was the topology of the network
that induced a more local 1-median subgraph and attenuated the
attraction force towards the optimal. With LOM, this locality is
inherently imposed a-priori by the method, with similarly negative
results.

On the other hand, the cDSMA heuristic seeks to choose the most
``appropriate'' candidate hosts, capable of leading the service
fast to preferable/cost-minimizing locations, no matter what the
shape/radius of the emerging $G^{i}$ neighborhood would be.
%When
%-occasionally- LOM reaches the optimal host, it does so with more
%migration hops than ALGNAME \textbf{(pointer to table)}.
Whereas, in a couple of cases that both approaches are trapped to
a suboptimal place, e.g., Dataset 33 and $D_{gen}$=7, LOM needs
three hops to get there, whereas cDSMA aborts after one hop.

This capability of cDSMA to make longer migration hops and
accelerate its convergence to the (sub)optimal service location
has another positive effect: the migration hop-count remains
largely independent of the service generation host. This means
that the mechanism does not favor nodes according to their
proximity to the service demand and/or network topology hot spots,
inducing a less dramatic yet welcome notion of fairness in the
performance different network users get.

%% file: related.tex
\section{Related work}
\label{Related_work}

The problem of service placement has been predominantly treated as
an instance of the broader family of (metric) facility location
(FL) problems, which have found many different applications in
areas as diverse as transportation networks and distributed
computing. (Un)Capacitated FL problems is probably the most
popular problem variant, where the objective is to minimize the
combined cost of opening a facility and serving its clients and
the number of facilities is not a priori bounded. The problem we
address instead is an instance of k-median, $k=1$ in our case,
problems, where no opening cost exists and the operational
facilities cannot exceed $k$. Both problems are NP-hard for
general topologies~\cite{mirchandani1990,kariv1979}; thus, various
approximations commonly requiring exact knowledge about their
inputs, have been proposed to address
them~\cite{alg_k_Med,EvaTardos}.

%ttractive yet demanding research thread mainly due to its wide
%applicability ranging from transportation networks to distributed
%computing. Theoretically, it is related to two optimization
%problems, namely, the \textit{k}-median (an istance of which, was
%addressed herein) and the \textit{facility location} problem; the
%latter, given a cost metric associated with the opening of each
%service facility, involves the decision on both the number and
%locations of the facilities to be placed in a way that the
%cummulative service and opening cost would be minimized. Both
%problems are NP-hard for general
%topologies~\cite{mirchandani1990,kariv1979} and, thus, various
%approximations commonly requiring exact knowledge about their
%inputs, have been proposed to address
%them~\cite{alg_k_Med,EvaTardos}.

The proposed approaches are typically categorized to centralized
and distributed. The applicability of centralized solutions to
large-scale data networks is severely undermined by the need for
centralized decision-making and collection of global information
about service demand and topology. In particular when this
information varies dynamically, as with mobile networks,
distributed solutions become mandatory and have recently received
%the focus is on the distributed-fashion solutions; moreover, when
%limitations by possible dynamicity of mobile networks are imposed,
%devising distributed solution has received renewed
renewed attention~\cite{1372881}.

One recently initiated research thread relates to the
\emph{approximability} of distributed approaches to the facility
location problem. Moscibroda and Wattenhofer in
\cite{Moscibroda05} draw on a primal-dual approach earlier devised
by Jain and Vazirani in \cite{JainVazi01}, to derive a distributed
algorithm that trades-off the approximation ratio with the
communication overhead under the assumption of $O(logn)$ bits
message size, where $n$ the number of clients. More recently,
Pandit and Pemmaraju have derived an alterative distributed
algorithm that compares favorably with the one in
\cite{Moscibroda05} in resolving the same trade-off
\cite{Pandit09}.
%me significant effort has been lately made towards the facility
%location distributed approximations yielding successively better
%ratios during the last few years [\textbf{latest ref: Return of
%the Primal-Dual???}]. In~\cite{citeulike:2391422} the facility
%location problem is addressed in a constant number of
%communication rounds. The method iteratively approximates the
%parameters of a primal-dual problem and then a distributed
%randomized rounding technique is applied to map the fractional
%results back to the original problem. Still, the algorithm
%requires some global view to compute and distribute a coefficient
%to all nodes, before execution, as well as communication among all
%relevant clients and facilities in each round. Furthermore, the
%approximation guarantees provided by such algorithms,   for some
%instances reaching almost the xx factor [??], can be typically
%outperformed by practical heuristic solutions.

Although the approximability studies can provide provable bounds
for the run time and the obtainable quality of the solutions, they
are typically outperformed by less mathematically rigorous yet
practical heuristic solutions. Common to most of them is the
service \emph{migration} from the generator host towards its
optimal location through a number of locally-determined hops that
delineate a cost-decreasing path. What changes is the way
decisions are made. Oikonomou \etal in ~\cite{OikonomouSX08}
exploit the shortest-path tree structures that are induced on the
network graph by the routing protocol operation to estimate upper
bounds for the aggregate cost in case service migrates at the
1-hop neighbors. Migration hops are therefore one physical hop
long and this decelerates the migration process, especially in
larger networks. Our algorithm resolves more efficiently the
trade-off between convergence speed and accuracy; in fact, cDSMA
maintains consistently high convergence speed over the real-world
topologies while achieving very-close-to-optimal placements.

%Having the service demand of all nodes belonging to subtrees
%routed at the current service host, aggregated at the host's first
%neighbours, the authors demonstrate (by analytical means) policies
%to migrate the service following (untill the end) a cost metric
%decreasing path, when some criteria referring to the number of
%facilities, link weights and topology, hold. A combination of
%policies is also applied to avoid one-hop tentative moves that
%otherwise introduce overhead. Simulation, albeit conducted only on
%synthetic topologies, shows that combining policies can lead the
%service to cost-minimizing locations but compared to our approch
%it yields poor results in terms of both accuracy and convergence
%speed.

Even closer to our work is the upcoming paper of Smaragdakis
\etal~\cite{SLOSB-DSMSISD-10}. They reduce the original k-median
problem in multiple smaller-scale 1-median problems solved within
an area of \textit{r-hops} from the current location of each
service facility. Compared to cDSMA, the area over which they
search for candidate next service hosts and upon which they map
demand from the ``outer'' nodes is the \textit{r}-hop
neighbourhood of the current service location.
In~\ref{subsec:positioning} we have discussed in detail how cDSMA
compares with a similar, local-search oriented approach.

Finally, cDSMA is an instance of a mechanism, where insights from
Complex Network Theory help improve the performance of a network
operation (\emph{here:} service migration and optimal placement)
significantly. Two more such examples have been reported in the
area of Delay Tolerant Networks, where CNA has inspired the
derivation of new routing protocols that, when correctly tuned,
can improve performance significantly over more naive approaches
\cite{SimBet, BubbleRap}.

%% file: discussion.tex
\section{Discussion-conclusions}\label{sec:discussion}

%\textbf{Main points:} \\
%a) Summarize results, draw hints for appropriateness of the algorithm.
%Consider scenarios that favor it. \\
%b) Discuss what information is needed and how could it be
%collected \\
%c) Discuss assumptions and how could they be relaxed (e.g.,
%shortest path rou ting and RW-based or k-centrality) \\

Networked communication becomes more and more user-oriented. After
the success of user-generated content, user-oriented service
creation emerges as a new paradigm that will let individual users
generate and make available services at minimum programming
effort. Scalable distributed service migration mechanisms will be
key to the successful proliferation of the paradigm.

We have mimicked earlier research work in treating the service
placement problem in the general context of facility location
problems. We have departed from it in exploiting complex network
analysis for coming up with a scalable distributed service
migration mechanism. We introduced a metric, weighted Conditional
Betweeness Centrality ($wCBC$) that captures the topological
centrality and demand aggregation capacity of individual nodes.
The metric is used to select a small subset of significant nodes
for solving the 1-median problem as well as easily map the demand
of the remaining nodes on this subset. The service facilities
migrate in the network towards the (sub)optimal location along a
cost-decreasing path determined iteratively at the few
intermediate service host nodes.

Both the network topology and spatial dynamics of service demand
affect the accuracy and the convergence speed of the algorithm,
giving rise to stronger/lighter attraction forces that drag the
migrating service facilities towards the optimal location. In
general, the higher the asymmetry in either of the two, the better
the performance of the algorithm. The exhaustive evaluation of our
algorithm on real-world topologies suggests that very good
accuracy can be obtained when solving the 1-median problem with a
very small number, in the order of ten, of nodes with the highest
$wCBC$ scores. The result is insensitive to the network size and
diameter and the asymmetry of demand distribution, hinting that
real-world topologies have enough asymmetry to yield good
performance of the algorithm. Moreover, the algorithm outperforms
locality-oriented service migration and its accuracy and
convergence speed are not dependent on the position of the service
generation.

The proposed mechanism is highly decentralized; all
nodes-candidates to host a service share the decision-making
process for optimally placing the service in the network. It is
also scalable in that it copes with the computational burden
related to the solution of the $1$-median problem; this may become
a difficult task for large-scale networks, especially when changes
in the service demand characteristics call for its repeated
execution. Nevertheless, topological and demand information still
needs to propagate in the network. For small-size networks,
topological information may become available through the operation
of a link-state routing protocol that distributes and uses global
topology information. For larger-scale networks, one way to
acquire topology information would be through the deployment of
some source-routing or path-switching protocol that carries
information about the path it traverses on its headers.
Information about the interest in services, on the other hand, may
need more effort. Users increasingly subscribe to social
networking sites and, sometimes consciously, give information
about their interests and preferences. Profile-building mechanisms
are components of peer-to-peer protocols as well; our mechanism
will also be ultimately part of such a protocol.

Our problem formulation and the $wCBC$ metric we introduced for
harnessing the computational burden of the $1$-median problem
solution assume that the network exercises minimum hop count
routing. Although minimum hop count routing is both simple and
popular, network traffic engineering requires more elaborate
routing solutions such as load-balancing/load adaptive routing
\cite{trafficengineering}. We could generalize our treatment of
the service migration problem to address these cases. First of
all, the (conditional) betweenness centrality factor in the $wCBC$
metric definition is inherently flexible in that it considers
shortest paths. Different routing metrics can be accommodated
through changing the context of (shortest) path. For example, we
could consider weighted graphs, where link weights may represent
link capacities or propagation delays. A more substantial change
in the metric would be to replace the shortest-path betweenness
centrality with alternative $BC$ definitions: the
\emph{random-walk betweeness centrality} \cite{Newman200539},
which would resemble more a probabilistic, traffic demand
oblivious routing implementation, or the \emph{k-betweenness
centrality} \cite{k-BC09}, which is closer to some short of
multipath routing, even if it does not enforce independent,
link/node disjoint paths. On the other hand, to accommodate
other-than-min-hop-count routing in the content access leg, we
would need a fundamental adaptation of the $1$-median problem
formulation.